\documentclass[preprint,review,12pt,3p,authoryear]{elsarticle}
\usepackage{graphics}
\usepackage{graphicx}
\usepackage{epsfig}
\usepackage{mathrsfs, amsthm,amsfonts,amsmath,graphics,graphicx,bm,enumerate}
\usepackage{amssymb,amsfonts,bm}
\usepackage{amsthm}
\usepackage{amsmath}
\allowdisplaybreaks[4]
\usepackage{mathtools}
\usepackage[ruled]{algorithm2e}
\usepackage{url}
\usepackage{natbib}
\usepackage{perpage} %the perpage package

\usepackage{lineno}

\usepackage[hidelinks]{hyperref}
\hypersetup{
colorlinks = true,
linkcolor=blue,
citecolor=blue,
urlcolor=blue 
}

\usepackage{textcomp}
\usepackage{booktabs,multirow, multicol}
\newtheorem{thy}{Theorem}

\journal{European Journal of Operational Research }

\begin{document}

\begin{frontmatter}
\title{Online Bayesian prediction of remaining useful life for gamma degradation process under conjugate priors}

%\tnotetext[label0]{This is only an example}
\author[label1,label2]{Ancha Xu\corref{cor1}}
%\ead{mailto:xuancha2011@aliyun.com}

\address[label1]{Department of Statistics, Zhejiang Gongshang University, Zhejiang, China}
\address[label2]{Collaborative Innovation Center of Statistical Data Engineering, Technology \& Application\\ Zhejiang Gongshang University, Zhejiang, China}
\cortext[cor1]{Corresponding author}
\ead{xuancha2011@aliyun.com}

\begin{abstract}
Gamma process has been extensively used to model monotone degradation data. Statistical
inference for the gamma process is difficult due to the complex parameter structure involved in the likelihood function.
In this paper, we derive a conjugate prior for the homogeneous gamma process, and some properties of the prior distribution
are explored.
Three algorithms (Gibbs sampling, discrete grid sampling, and sampling importance resampling) are well designed to generate posterior samples of the model parameters, which can greatly lessen the challenge of posterior inference. Simulation studies show that the proposed
algorithms have high computational efficiency and estimation precision. The
conjugate prior is then extended to the case of the gamma process with heterogeneous effects.
With this conjugate structure,
the posterior distribution of the parameters can be updated recursively, and
an efficient online algorithm is  developed to predict remaining useful life of multiple systems.  The effectiveness of the proposed
online algorithm is illustrated by two real cases.

\end{abstract}

\begin{keyword}
%% keywords here, in the form: keyword \sep keyword
Reliability \sep Heterogeneity \sep Gibbs sampling \sep Sampling importance resampling \sep Remaining useful life.

\end{keyword}	
\end{frontmatter}

\newpage
\section{Introduction}\label{sec_intro}

Modern systems are often designed with high-quality standards, such as wind turbines \citep{xuleili2020}, bearings in high-speed trains \citep{silizha2019},  plasma display panels \citep{chapul2016}, lithium-ion batteries \citep{xulichen2016}, etc. In the constraint time, it is difficult to get failure information of these systems from life testing, which poses a significant challenge to manufacturing firms. Fortunately, the ageing failures of these systems are usually attributed to some underlying performance characteristics (PCs), for instance, crack size of  the bearing, lumen output of the light-emitting diode, a lithium-ion battery's capacity, etc. Degradation of PC  accumulates over time and eventually reaches a predetermined threshold.
The first hitting time to the threshold can be viewed as the system's lifetime.
The link between degradation and system failure provides
a promising way to assess the reliability of highly reliable systems,
as it is possible to estimate the failure time distribution  through a certain degradation-based model.
Using the same experimentation time,
degradation tests  have been demonstrated to provide more life information
than traditional life testing. As a result,  analysis of degradation data is expected to estimate the system's
lifetime distribution more accurately, which has also been demonstrated as  an effective way for reliability assessment.

The existing degradation models are mainly composed of two categories: stochastic process models and
general path models. The distinctions between
the two types of models have been well addressed  by  \cite{yexie2015}.
In real-world applications, stochastic process models are more widely utilized because of  their mathematical properties and  physical explanations. As a special
stochastic degradation model, the gamma process can be interpreted as
the limit of a compound Poisson process
with the jump size following a specific distribution, and  is often adopted when the PC deterioration is strictly monotone. The gamma process as a class of degradation models was first introduced by \cite{sing1995}.
Several extensions of the gamma process that take into account covariates, heterogeneous effects,
measurement errors, and multistage degradation have been well studied over the last two decades. For example,
\cite{bagnik2001} modelled the gamma process with covariates by using the method of additive accumulation of damages. \cite{parpad2006}  proposed an accelerated gamma degradation
model with the assumption that the shape parameter is a function of covariates. \cite{lintsubal2015} presented an accelerated gamma degradation model with
bounded constraint.  When there was unit-to-unit variation,
\cite{lawcro2004} considered the gamma process with heterogeneous effects, where the scale
parameter was assumed to be a random variable with gamma distribution.
\cite{wang2008} proposed a pseudo-likelihood method to estimate the parameters under
non-homogeneous gamma process model with random effects.
\cite{wanwanhon2022} developed a generalized inference method
for the gamma process with random effects, which can generate accurate interval estimates for
the model parameters. 
When 
the degradation process is imperfectly inspected,
the measurement errors  are non-ignorable, and Gaussian  distributed noise can be included. Then independence among the  degradation increments
does not hold, which makes the parameter estimation intractable.  \cite{hazpanman2020} proposed
approximate Bayesian computation method to handle this problem, and
\cite{espmelcas2022}  combined particle filter and an expectation-maximization algorithm
to obtain the parameter estimation. For some special systems, due to physical or chemical changes, the degradation path of PC may
exhibit two phases,  for example,  the luminosity of organic light-emitting diode \citep{wantanbae2018},
the capacity of
lithium-ion batteries \citep{linlincab2021}.  \cite{linngtsu2019} considered
two-phase degradation models under the gamma process, as well as
Bayesian and likelihood methods for estimating the model parameters.
In addition to being a model of degradation, the gamma process also serves as  a powerful model
in other fields, such as statistical process control \citep{hsupeawu2008,chenye2018}, maintenance \citep{liupanwan2021},  sports science \citep{sonshi2020},  etc.

Another goal of modeling degradation data is to predict the remaining useful life (RUL) of the system.
With the development of sensor technology, the degradation of PC can be monitored in real-time, and the RUL is predictable regularly. 
The predicted RUL can then be timely used
to support  condition-based maintenance.  In the case of degradation-based online RUL prediction, the degradation models and statistical inference methods of parameter estimation are the two key components. For the Wiener-based degradation models,  the Kalman filter or methods based on the Kalman filter are often adopted
to predict RUL online. These methods are capable of  achieving closed-form online RUL prediction with no requirement on historical data storage for linear degradation models \citep{siwanghu2013,wangtsui2018,zhangsihu2018}. However, the implementation
of the Kalman filter and its related methods is founded on the Gaussian distribution, which
restricts their applications. 
For gamma degradation models, \cite{paroissin2017} and \cite{xushen2018} developed
recursive linear estimators of the mean and variance of the gamma process, while 
the RUL prediction as well as its interval estimation can not be obtained by the same techniques. 
The current offline methods, such as Bayesian and likelihood-based techniques \citep{wang2008, linngtsu2019, wanwanhon2022}, are based on the entire set of data. When new observations are available, 
statistical analysis  needs to be re-conducted for the updated dataset. As the sample size grows,  
data storage and analysis  based on  these methods will become challenging.
Regarding this, an efficient method with low computational requirements is necessary for
online RUL prediction under gamma process. A promising solution to this problem is
using conjugate priors for the gamma process. By the nice properties of conjugate priors,
recursive Bayesian analysis is possible, and the online RUL prediction can also be realized efficiently.
The problem arises from the 
fact that the conjugate prior distribution is complicated.
This makes dealing with posterior inference
difficult. In light of this, we develop three algorithms to simulate random numbers from the
posterior distribution, which greatly reduces the computational burden of posterior inference.  
We then propose an online RUL prediction algorithm that exploits the advantages of conjugate priors and maintains the tractability of the closed-form update. Thus, it guarantees fast online RUL prediction of multiple systems with minimal computational power requirements. 

The remainder of this paper is organized as follows. In Section \ref{sec:prior}, we propose a class of prior distributions for the
gamma process, and investigate some properties of the priors. Three algorithms are presented to generate posterior samples based on conjugate priors in Section \ref{sec:ps}. Simulation studies
are carried out to compare the three algorithms in terms of estimation accuracy and computational
efficiency in Section \ref{sec:sim}. The conjugate priors are extended to the case of gamma process with heterogeneous effects in Section \ref{sec:heter}. 
An online RUL prediction algorithm based on conjugate priors is explored in Section \ref{sec:rul}. 
Section \ref{sec:case} demonstrates the online RUL prediction algorithm in two real cases.
Section \ref{sec:conc} concludes the paper.

\section{Conjugate prior}	
\label{sec:prior}	
If a stochastic process $\{\mathcal{Y}(t),t\ge0\}$ satisfies the following properties:
\begin{enumerate}
	\item[i)]  $\mathcal{Y}(0)=0$ with probability 1,
	\item[ii)] $\{\mathcal{Y}(t),t\ge0\}$ has stationary and  independent
	increments,
	\item[iii)] the increment $\Delta Y_t=\mathcal{Y}(t)-\mathcal{Y}(s)$ follows
	gamma distribution
	( $Ga(\alpha(t-s),\beta)$) with probability density function (PDF)
	$$f(y| \alpha, \beta)=\frac{\beta^{\alpha(t-s)} y^{\alpha(t-s)-1}}{\Gamma(\alpha(t-s))}
	\exp\left\{-\beta y \right\}, t>s,$$
	where $\Gamma(\cdot)$ denotes the gamma function, $\alpha$ and $\beta$ are positive parameters,
\end{enumerate}
then it is  called homogeneous gamma process, denoted by  $\{\mathcal{Y}(t),t\ge0\}\sim\mathcal{GP}(\alpha t,\beta)$.	

Gamma process is widely used to describe the deterioration path of some systems' PC.
Let $\mathbb{C}$ denote the threshold level of a system's PC. Then the lifetime of the system is
defined as  $\mathcal{T}=\inf\{t|\mathcal{Y}(t)\ge\mathbb{C}\}$. For gamma degradation process $\mathcal{GP}(\alpha t,\beta)$, the cumulative distribution function (CDF) of $\mathcal{T}$ is
\begin{equation}
	\begin{aligned}
		F_{\mathcal{T}}(t|\alpha,\beta)&=P(\mathcal{T}< t)=P(\mathcal{Y}(t)>\mathbb{C})=\frac{\Psi(\beta\mathbb{C},\alpha t)}{\Gamma(\alpha t)},
	\end{aligned}
\end{equation}
where $\Psi(k,\alpha)$ is the incomplete gamma function defined by $\Psi(k,\alpha)=\int_{k}^{\infty}x^{\alpha-1}\exp(-x)\text{d}x$.
Although $F_{\mathcal{T}}(t|\alpha,\beta)$ has an analytic form, the PDF of $\mathcal{T}$ is too complicated to
be applied in practice.
\cite{parpad2005} recommended a two-parameter  Birnbaum-Saunders distribution $BS(\alpha^\ast,\beta^\ast)$ with CDF $\Phi\left(\frac{1}{\alpha^\ast}\left[\sqrt{\frac{t}{\beta^\ast}}-\sqrt{\frac{\beta^\ast}{t}}\right]\right)$
to approximate $F_{\mathcal{T}}(t|\alpha,\beta)$, where $\alpha^\ast=\sqrt{\frac{1}{\beta\mathbb{C}}}$ and
$\beta^\ast=\frac{\beta \mathbb{C}}{\alpha}$, $\Phi(\cdot)$ is the
CDF of standard normal distribution. Therefore, the mean-time-to-failure (MTTF) of the system can be approximated by
$\beta^\ast\left(1+\frac{\left(\alpha^\ast\right)^2}{2}\right)=\frac{1+2\beta\mathbb{C}}{2\alpha}$.

Assume that the degradation path of system's PC follows gamma process $\mathcal{GP}(\alpha t,\beta)$. A total of
$n$ systems from population are randomly selected and  tested.
The measurement time epochs are $T_1<T_2<\cdots<T_m$, and the corresponding degradation value of the $i$-th system at time epoch $T_j$ is $Y_{ij}$, $i=1,\dots,n,$ $j=1,\dots,m$.
Let $y_{ij}=Y_{ij}-Y_{ij-1}$ and $t_j=T_j-T_{j-1}$, where $Y_{i0}=0$ and $T_0=0$, $i=1,\dots,n,$ $j=1,\dots,m$.
Denote the observed data as $\bm{y}=\{y_{ij}, i=1,\dots,n,~j=1,\dots,m\}$. According to the property (iii) of gamma process,   $y_{ij}\sim Ga(\alpha t_j,\beta)$. Then based on the data  $\bm{y}$, the likelihood function of $\alpha$ and $\beta$ is
\begin{equation}\label{like0}
	\begin{aligned}
		L( \bm{y}|\alpha,\beta)&=\prod_{i=1}^{n}\prod_{j=1}^{m}\dfrac{\beta^{\alpha t_j}}{\Gamma(\alpha t_j)}
		y_{ij}^{\alpha t_j-1}\exp\{-\beta y_{ij}\}\\
		&\propto\dfrac{\beta^{nT_m\alpha}\left[\prod_{i=1}^{n}\prod_{j=1}^{m}y_{ij}^{t_{j}}\right]^\alpha}
		{\left[\Gamma(\alpha t_j)\right]^n}\exp\left\{-\beta \sum_{i=1}^{n}\sum_{j=1}^{m}y_{ij}\right\}\\
		&\propto \dfrac{\beta^{nm\overline{T}_m\alpha}}{\left[\prod_{j=1}^{m}\left(\Gamma(\alpha t_j)\right)^{1/m}\right]^{mn}}
		\left[\prod_{i=1}^{n}\prod_{j=1}^{m}y_{ij}^{\frac{t_{j}}{nm\overline{T}_m}}\right]^{nm\overline{T}_m\alpha}\exp\{-mn\bar{y}_a\beta\},
	\end{aligned}
\end{equation}
where $\overline{T}_m=\frac{T_m}{m}$ and  $\bar{y}_a=\frac{1}{mn}\sum_{i=1}^{n}\sum_{j=1}^{m}y_{ij}$
is the arithmetic mean of increments.

\begin{thy}\label{th1}
	Based on likelihood function \eqref{like0}, a conjugate prior of  $\alpha$ and $\beta$ is
	\begin{equation}\label{con0}
		\pi(\alpha,\beta)=	C \cdot\dfrac{(\beta\omega)^{\delta \overline{T}_m\alpha}}{\left[\prod_{j=1}^{m}\left(\Gamma(\alpha t_j)\right)^{1/m}\right]^{\delta}}
		\exp\{-\delta\lambda\beta\},
	\end{equation}
	where $C$ is a normalized constant, $\delta$, $\omega$ and $\lambda$ are hyperparameters with nonnegative values, which describe kurtosis, shape and scale of the distribution, respectively.
\end{thy}

\begin{proof}[Proof:]
	Based on the likelihood function \eqref{like0} and the prior \eqref{con0}, the joint posterior
	density of $\alpha$ and $\beta$ is
	\begin{equation}\label{post0}
		\begin{aligned}
			\pi(\alpha,\beta|\bm{y})&\propto 	L( \bm{y}|\alpha,\beta)\pi(\alpha,\beta)\\
			& \propto \dfrac{\beta^{(nm+\delta)\overline{T}_m\alpha}\omega^{\delta \overline{T}_m\alpha}}{\left[\prod_{j=1}^{m}\left(\Gamma(\alpha t_j)\right)^{1/m}\right]^{mn+\delta}}
			\left[\prod_{i=1}^{n}\prod_{j=1}^{m}y_{ij}^{\frac{t_{j}}{nm\overline{T}_m}}\right]^{nm\overline{T}_m\alpha}\exp\{-(mn\bar{y}_a+\delta\lambda)\beta\}\\
			&\propto \dfrac{(\beta\omega_p)^{\delta_p\overline{T}_m\alpha}}{\left[\prod_{j=1}^{m}\left(\Gamma(\alpha t_j)\right)^{1/m}\right]^{\delta_p}}
			\exp\{-\delta_p\lambda_p\beta\},
		\end{aligned}
	\end{equation}
	where $\delta_p=mn+\delta$, $\omega_p=\omega^{\frac{\delta}{mn+\delta}}\left[\prod_{i=1}^{n}\prod_{j=1}^{m}y_{ij}^{\frac{t_{j}}{nm\overline{T}_m}}\right]^{\frac{mn}{mn+\delta}}$ and $\lambda_p=\frac{mn}{mn+\delta}\bar{y}_a+\frac{\delta}{mn+\delta}\lambda$. Thus, $\pi(\alpha,\beta)$
	and $\pi(\alpha,\beta|\bm{y})$ are from the same distribution family.
\end{proof}
The conjugate prior $\pi(\alpha,\beta)$ depends on measurement time epochs,
and the form of $\pi(\alpha,\beta)$ seems to be complicated. However, it will be
beneficial to take another look at $\pi(\alpha,\beta)$:
\begin{equation}\label{deco0}
	\begin{aligned}
		\pi(\alpha,\beta)&=\pi(\beta|\alpha)\pi(\alpha)
		\propto
		\dfrac{(\delta\lambda)^{\delta \overline{T}_m\alpha+1}}{\Gamma\left(1+\delta \overline{T}_m\alpha\right)}
		\beta^{\delta \overline{T}_m\alpha}\exp\{-\delta\lambda\beta\}\cdot
		\dfrac{\left(\dfrac{\omega}{\delta\lambda}\right)^{\delta \overline{T}_m\alpha}\Gamma\left(1+\delta \overline{T}_m\alpha\right)}{\left[\prod_{j=1}^{m}\left(\Gamma(\alpha t_j)\right)^{1/m}\right]^{\delta}}.
	\end{aligned}
\end{equation}
Given $\alpha$, the conditional prior $\pi(\beta|\alpha)$ is gamma distribution $Ga(\delta \overline{T}_m\alpha+1,\delta\lambda)$.
Thus, the mode and variance of $\pi(\beta|\alpha)$ are $\overline{T}_m\alpha/\lambda$ and $(\delta \overline{T}_m\alpha+1)/(\delta\lambda)^2$, respectively. The hyperparameter
$\lambda$ is a standard scale parameter, while
the hyperparameter $\delta$ affects only  the variance rather than the mode of the conditional prior when $\alpha$ is given.
The curve of $\pi(\beta|\alpha)$ is concentrated around the mode for large values of $\delta$. In other words,
$\delta$ displays the kurtosis of $\pi(\beta|\alpha)$. We call $\delta$ the kurtosis parameter. The marginal prior
of $\alpha$ is proportional to
\begin{equation*}
	h(\alpha)=	\dfrac{\left(\dfrac{\omega}{\delta\lambda}\right)^{\delta \overline{T}_m\alpha}\Gamma\left(1+\delta \overline{T}_m\alpha\right)}{\left[\prod_{j=1}^{m}\left(\Gamma(\alpha t_j)\right)^{1/m}\right]^{\delta}}.
\end{equation*}
Using Stirling's formula and as $\alpha\rightarrow\infty$,
\begin{equation}\label{appro0}
	h(\alpha)\equiv O\left(\alpha^{(\delta+1)/2}\exp\left\{-\alpha\delta \overline{T}_m
	\left[\log\left(\frac{\lambda}{\omega}\right)+\log\left(\frac{\prod_{j=1}^{m}t_j^{t_j/T_m}}{\overline{T}_m}\right)\right]\right\}\right),
\end{equation}
where $h(\alpha)=O(g(\alpha))$ represents that $h(\alpha)$ and $g(\alpha)$ have the same order. It can be shown that $\log\left(\frac{\prod_{j=1}^{m}t_j^{t_j/T_m}}{\overline{T}_m}\right)\ge 0$ (See the proof in \ref{ap2}).
Thus, to guarantee that $\pi(\alpha)$ is a proper PDF, the condition of $\omega<\lambda$ should be ensured when determining the conjugate prior $\pi(\alpha,\beta)$. According to \eqref{appro0}, we know that the tail of $\pi(\alpha)$
behaves similar to that of gamma distribution $Ga\left(\frac{\delta+3}{2},\delta \overline{T}_m
\left[\log\left(\frac{\lambda}{\omega}\right)+\log\left(\frac{\prod_{j=1}^{m}t_j^{t_j/T_m}}{\overline{T}_m}\right)\right]\right)$.  $\omega$ behaves as a scale parameter in the $\pi(\alpha)$, which further affects the shape of $\pi(\beta|\alpha)$.  Thus, $\omega$ is called the shape parameter. Because of the gamma conditional prior $\pi(\beta|\alpha)$ and  tail property of $\pi(\alpha)$, the conjugate prior $\pi(\alpha,\beta)$ is referred to as approximated-gamma-gamma (AGG) distribution, denoted as $AGG(\delta,\omega,\lambda)$.

Figure \ref{config0} shows the function graphs and contours
of $\pi(\alpha,\beta)$ with various values of $(\delta,\omega,\lambda)$ when $t_j=1$, $j=1,\dots,m$. The top two subfigures in Figure \ref{config0} are the function graph and contour of $\pi(\alpha,\beta)$ with $\delta=2$, $\omega=0.5$, and $\lambda=1.5$, which is set as a benchmark. As can be seen in Figure \ref{config0}, when the value of $\delta$ is increased  to 5 and the other two hyperparameters are fixed,   the position of the mode is nearly identical, however, the contour is more concentrated around the mode. Increasing the value of $\lambda$ has similar
phenomena, while the mode is altered. A larger $\omega$ will increase the divergence of $\pi(\alpha,\beta)$
and also change the position of its mode. Figure \ref{config0} demonstrates
the influence of $\delta$, $\omega$ and $\lambda$ on the shape of the AGG distribution, which serves as a guide for selecting hyperparameter values based on beliefs of prior information.

{\bf Remark 1:} When the measurements are equally spaced, namely, the lag between two measurement time epochs $t_j= l$,
$\pi(\alpha,\beta)$ has a much simpler form:
\begin{equation}\label{con1}
	\pi(\alpha,\beta)=C \dfrac{(\beta\omega)^{\delta l\alpha}}{\left[\Gamma(l\alpha)\right]^{\delta}}
	\exp\{-\delta\lambda\beta\}.
\end{equation}
While $l=1$, $\pi(\alpha,\beta)$ is reduced to be a conjugate prior for gamma distribution $Ga(\alpha,\beta)$ \citep{dams1975}.

{\bf Remark 2:} The values of hyperparameters can be established based on the amount of prior information. As shown in Figure
\ref{config0}, large $\delta$, small $\omega$, or large $\lambda$ will lead to the small variance of $(\alpha,\beta)$,
which corresponds to the case of strong prior information. In the case of little prior knowledge, one may choose
a small $\delta$, large $\omega$, or small $\lambda$. In practical applications, we recommend using
$\delta$ to adjust the belief of prior information. As an example, in \eqref{post0}, we know that the posterior distribution
of $\alpha$ and $\beta$ is $AGG(\delta_p,\omega_p,\lambda_p)$. Special choices for $\omega$ and $\lambda$ can be
\begin{equation}\label{hyper}
	\omega=\prod_{i=1}^{n}\prod_{j=1}^{m}y_{ij}^{\frac{t_{j}}{nm\overline{T}_m}},
	~\lambda=\bar{y}_a,
\end{equation}
which are related to the observed data. Data-driven priors are not uncommon in statistics. For instance, Zellner's prior for regression coefficients \citep{zellner1986},  informative prior for threshold parameter \citep{halwan2005},
reference prior for linear degradation path model  \citep{xutang2012}, etc. These priors have
been demonstrated to be effective in practice. For \eqref{hyper},
several advantages should be indicated: (I)
The condition $\omega<\lambda$ for proper conjugate prior will be automatically satisfied in this setting.
%When measurements are equally spaced, $\omega$ is the geometric mean of $y_{ij}s$.
(II) $\omega$ and $\lambda$ determine the mode position of $\pi(\alpha,\beta)$, and this choice makes use of data information to suggest a reasonable
mode position.
(III) In this setting, the hyperparameter $\delta$ behaves like the number of measurements.
The value of $\delta$ can be chosen according to measurement-equivalent of the amount of information, e.g., $\delta=1$ can be interpreted as the amount of prior information equivalent to that of a system taking one measurement; $\delta=0$ represents noninformative prior. Thus, $\delta$
represents the belief of mode position suggested by \eqref{hyper}.
In terms of these merits, we will utilize the automatic strategy  \eqref{hyper}
for specifying hyperparameter values in simulation studies and data analysis, which could
greatly simplify conjugate prior specification.

\begin{figure}[htp]
	\vspace{-15mm}
	\centering
	\includegraphics[width=15cm]{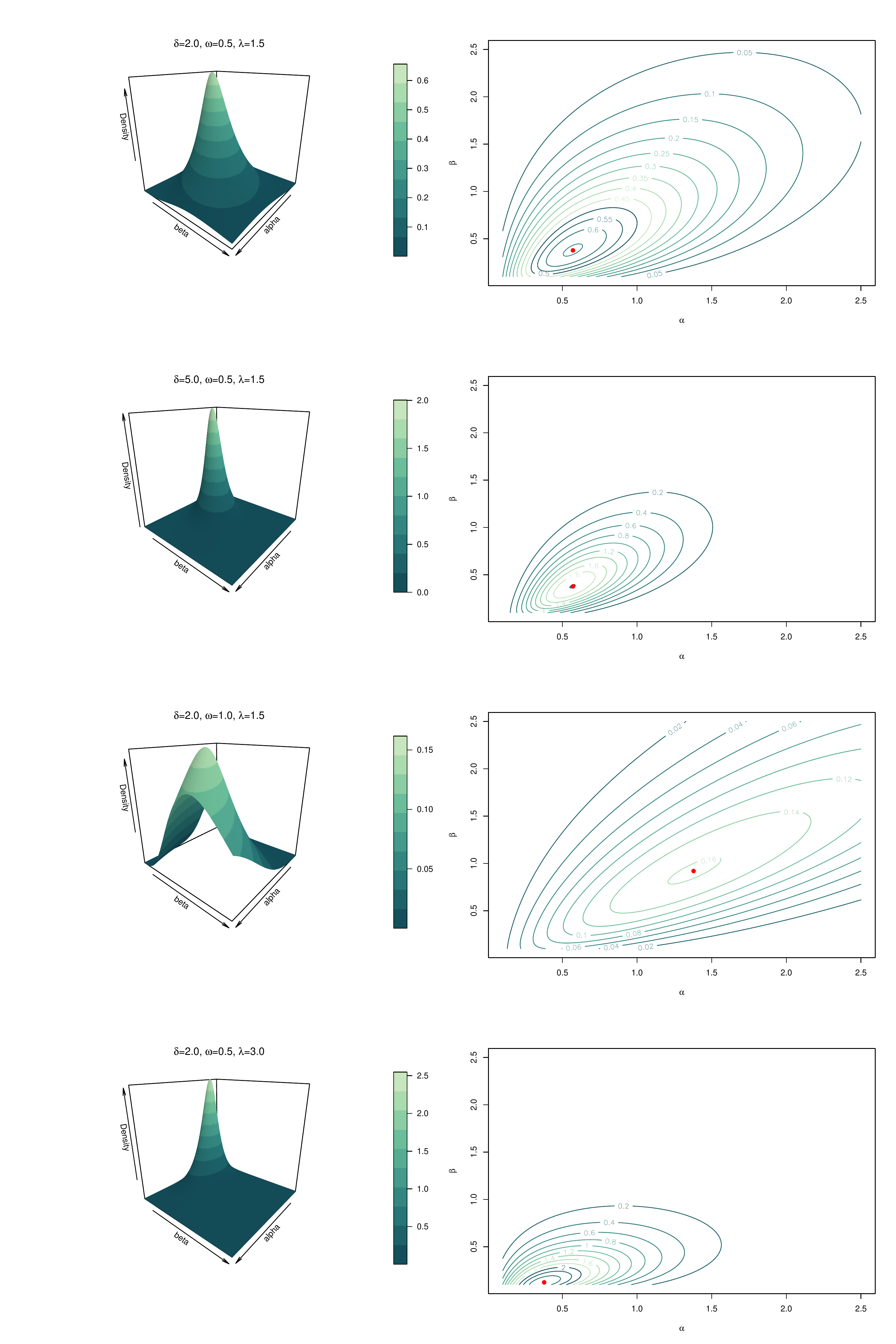}
	\caption{The figures on the left show the PDFs of conjugate priors and the figures on the right
		are the contours of their corresponding PDFs.}\label{config0}
\end{figure}

%For the sake of simplifying notations, we assume that the measurements are equally spaced. That is, the measurement time epoch $t_j= l$.

% \begin{equation}\label{like1}
% 	\begin{aligned}
% 		L( \bm{y}|\alpha,\beta)&=\prod_{i=1}^{n}\prod_{j=1}^{m}\dfrac{\beta^{\alpha l}}{\Gamma(\alpha l)}
% 		y_{ij}^{\alpha l-1}\exp\{-\beta y_{ij}\}\\
% 		&\propto \dfrac{\beta^{mnl\alpha}}{\left[\Gamma(\alpha l)\right]^{mn}}
% 		\bar{y}_g^{mnl\alpha}\exp\{-mn\bar{y}_a\beta\},
% 	\end{aligned}
% \end{equation}
% where $\bar{y}_g=\left[\prod_{i=1}^{n}\prod_{j=1}^{m}y_{ij}\right]^{1/(mn)}$ and
% $\bar{y}_a=\frac{1}{mn}\sum_{i=1}^{n}\sum_{j=1}^{m}y_{ij}$.

\section{Posterior sampling}
\label{sec:ps}
Since the posterior $\pi(\alpha,\beta|\bm{y})$  is not a regular bivariate distribution,
explicit computations of  Bayesian estimates for $\alpha$ and $\beta$ are not possible. Monte Carlo methods provide
an alternative method to do posterior inference for intractable posterior distributions. The fundamental idea behind Monte Carlo methods
is to generate random numbers from the joint posterior distribution of parameters and to obtain the point and interval estimates of the parameters or functions of parameters by the simulated numbers. In this section, we present three algorithms to simulate random numbers from AGG distribution. %Without loss of generality, we take $AGG(\delta,\omega,\lambda)$ as an example to
%illustrate the algorithms, because the joint posterior distribution of $\alpha$ and $\beta$ is also
% AGG distribution but with different hyperparameters.

The first algorithm is Gibbs sampling, a special Markov chain Monte Carlo algorithm. Gibbs sampling is implemented based on full conditional posterior densities $\pi(\beta|\alpha,\bm{y})$ and $\pi(\alpha|\beta,\bm{y})$. Similar to
\eqref{deco0}, we know that the full conditional posterior density $\pi(\beta|\alpha,\bm{y})$ is gamma distribution
$Ga(\delta_p\bar{T}_m\alpha,\delta_p\lambda_p)$, and that the full conditional posterior density $\pi(\alpha|\beta,\bm{y})$ is proportional to
$$\dfrac{(\beta\omega_p)^{\delta_p \overline{T}_m\alpha}}{\left[\prod_{j=1}^{m}\left(\Gamma(\alpha t_j)\right)^{1/m}\right]^{\delta_p}}.$$
It can be easily shown that $\pi(\alpha|\beta,\bm{y})$ is log-concave. Then adaptive rejection sampling (ARS) algorithm  can be utilized to generate random numbers from $\pi(\alpha|\beta,\bm{y})$ \citep{gilwil1992}.
After obtaining posterior samples of $\alpha$ and $\beta$,   any function of the parameters
$\eta=p(\alpha,\beta)$ (e.g., reliability of system, MTTF) can also be estimated. The procedure of posterior inference by Gibbs sampling is summarized in Algorithm \ref{algo1}.

\begin{algorithm}
	\caption{Gibbs sampling}\label{algo1}
	\LinesNumbered
	\KwIn{Observed data $\bm{y}$.\\}
	\KwOut{The point estimates and $100(1-\rho)\%$ credible intervals of $\alpha,\beta$ and $\eta$.}
	\BlankLine
	Initialize $\alpha$ and $\beta$ randomly (say, $\alpha_0$ and $\beta_0$), and compute $\delta_p$, $\omega_p$ and $\lambda_p$.
	
	\For{$k$ in $\{1,2,\dots,K_1\}$}{
		Generate $\beta_k$ from $Ga(\delta_p\bar{T}_m\alpha_{k-1},\delta_p\lambda_p)$;
		
		Generate $\alpha_k$ from $\pi(\alpha|\beta_k,\bm{y})$ by ARS algorithm;
		
		Compute $\eta_k=p(\alpha_k,\beta_k)$.
	}	
	Discard the first $B$ burn-in random numbers, and choose thinning interval $L$ to eliminate autocorrelation among posterior samples.
	
	Calculate the point  and $100(1-\rho)\%$ estimates for these parameters by posterior sample $\{(\alpha_k,\beta_k,\eta_k),k=B+1,B+L+1,B+2L+1,\dots,K_1\}$.	
\end{algorithm}

The second algorithm is discrete grid sampling (DGS).
The posterior $\pi(\alpha,\beta|\bm{y})$ can be decomposed into  $\pi(\beta|\alpha,\bm{y})\pi(\alpha|\bm{y})$, where $\pi(\beta|\alpha,\bm{y})$ is $Ga(\delta_p\bar{T}_m\alpha,\delta_p\lambda_p)$, and $\pi(\alpha|\bm{y})$ is proportional to
\begin{equation}\label{pmargi}
	h_p(\alpha)=	\dfrac{\left(\dfrac{\omega_p}{\delta_p\lambda_p}\right)^{\delta_p \overline{T}_m\alpha}\Gamma\left(1+\delta_p \overline{T}_m\alpha\right)}{\left[\prod_{j=1}^{m}\left(\Gamma(\alpha t_j)\right)^{1/m}\right]^{\delta_p}}.
\end{equation}
According to the decomposition, we know that once a random number $\alpha_0$ from $\pi(\alpha|\bm{y})$ is given,
random number of $\beta$ can be generated directly from $Ga(\delta_p\bar{T}_m\alpha_0,\delta_p\lambda_p)$. The main
difficulty arises from generating the random number of $\alpha$. Because of the complicated nature of $\pi(\alpha|\bm{y})$,
we use a particularly
simple simulation approach, approximating the marginal posterior distribution by a discrete distribution supported on a set of grid points, which provides sufficiently accurate inferences for $\alpha$. Firstly, an interval
$[A_1,A_2]$ that captures almost all the mass of $\pi(\alpha|\bm{y})$ is chosen,
which can be checked whether $\int_{A_1}^{A_2}\pi(\alpha|\bm{y})\text{d}\alpha$ is sufficiently close to 1.
Six-sigma rule can be utilized to construct a reasonable interval for $\alpha$. The procedure is summarized
below.

\vspace{4mm}
\setlength{\fboxrule}{1pt}\fbox{\shortstack[l]{
		1. Let $\tilde{\alpha}=\mathop{\arg\max}\limits_{\alpha} \log h_p(\alpha)$ and
		$I\left(\tilde{\alpha}\right)=-\dfrac{\partial^2\log h_p(\alpha)}{\partial \alpha^2}\biggl|_{\alpha=\tilde{\alpha}}$.
		\\
		2. According to Berger (1985), $\pi(\alpha|\bm{y})$ can be approximated by
		normal distribution \\
		~~~~$N\left(\tilde{\alpha},\tilde{\sigma}^2\right)$,
		where $\tilde{\sigma}=\sqrt{1/I\left(\tilde{\alpha}\right)}$.\\
		3. Let
		$A_1=\max\{0,\tilde{\alpha}-6\tilde{\sigma}\}$ and $A_2=\tilde{\alpha}+6\tilde{\sigma}$.
		Then, according to the property of normal\\
		~~~ distribution, we know that
		the probability that $\alpha$ falls into interval $[A_1,A_2]$ is almost 1.
	}
}

\vspace{3mm}
Given $[A_1,A_2]$,
we select $M$ grid points $\left\{A_1=\alpha^{(1)},\alpha^{(2)},\dots,A_2=\alpha^{(M)}\right\}$
in the interval with equally spaced, and compute the probability for each grid point by unnormalized posterior density
$h_p(\alpha)$:
\begin{equation}\label{disalp}
	P(\alpha=\alpha^{(s)})=\dfrac{h_p\left(\alpha^{(s)}\right)}{\sum_{i=1}^{M}h_p\left(\alpha^{(i)}\right)},~ s=1,\dots,M.
\end{equation}
The approximation precision can be guaranteed with sufficient large $M$. The use of discrete approximation reduces the difficulty of sampling from $\pi(\alpha|\bm{y})$ greatly, because
simulating random number from discrete distribution is straightforward by statistical software, for instance,
the function {\it sample()} in R language. The procedure of posterior inference by DGS is summarized in Algorithm \ref{algo2}.

\begin{algorithm}
	\caption{DGS}\label{algo2}
	\LinesNumbered
	\KwIn{Observed data $\bm{y}$.\\}
	\KwOut{The point estimates and $100(1-\rho)\%$ credible intervals of $\alpha,\beta$ and $\eta$.}
	\BlankLine
	Compute $\delta_p$, $\omega_p$ and $\lambda_p$.
	
	Determine interval $[A_1,A_2]$ by six-sigma rule.
	
	Choose $M$ grid points $\left\{A_1=\alpha^{(1)},\alpha^{(2)},\dots,A_2=\alpha^{(M)}\right\}$, and compute
	the probability for each grid \eqref{disalp}.
	
	\For{$k$ in $\{1,2,\dots,K_2\}$}{
		Generate $\alpha_k$ from discrete distribution \eqref{disalp};
		
		Generate $\beta_k$ from $Ga(\delta_p\bar{T}_m\alpha_{k},\delta_p\lambda_p)$;
		
		Compute $\eta_k=p(\alpha_k,\beta_k)$.
	}	
	Calculate the point and $100(1-\rho)\%$ interval estimates of these parameters by posterior sample $\{(\alpha_k,\beta_k,\eta_k),k=1,\dots,K_2\}$.		
\end{algorithm}

The distinction between the third algorithm and the second algorithm mainly lies in the method of generating posterior samples from $\pi(\alpha|\bm{y})$, in which sampling importance resampling (SIR) is adopted. In SIR, rather than sampling from $\pi(\alpha|\bm{y})$ directly, the sampling step is carried out from an instrumental distribution $g(\alpha)$. There is little restriction on the choice of $g(\alpha)$, which can be chosen from a set of  distributions
that can be easily simulated. However, the efficiency of SIR depends on how closely $g(\alpha)$ can imitate $\pi(\alpha|\bm{y})$, especially in the tails of the distribution. Similar to \eqref{appro0}, we know that
the tail of $\pi(\alpha|\bm{y})$
has the same order as that of gamma distribution with shape parameter $(\delta_p+3)/2$
and scale parameter
$\nu=\delta_p\overline{T}_m
\left[\log\left(\frac{\lambda_p}{\omega_p}\right)+\log\left(\frac{\prod_{j=1}^{m}t_j^{t_j/T_m}}{\overline{T}_m}\right)\right]$. Thus, we choose gamma
distribution $Ga(a,b)$ as instrumental distribution. The values of $a$ and $b$ can be determined as follows.

\vspace{4mm}
\setlength{\fboxrule}{1pt}\fbox{\shortstack[l]{
		1. Let $\tilde{\alpha}=\mathop{\arg\max}\limits_{\alpha} \log h_p(\alpha)$ and
		$I\left(\tilde{\alpha}\right)=-\dfrac{\partial^2\log h_p(\alpha)}{\partial \alpha^2}\biggl|_{\alpha=\tilde{\alpha}}$.
		\\
		2. Initialize $b$ as $b_0=\nu$
		and $a$ as $a_0=\tilde{\alpha}b_0$. The initialized step ensures that the mean of
		\\~~~~$Ga(a_0,b_0)$  is $\tilde{\alpha}$.
		\\
		3.  Compute the precision ratio $R=\frac{b_0^2/a_0}{I\left(\tilde{\alpha}\right)}$, and
		update $a=a_0/R$ and $b=b_0/R$.
		This step \\
		~~~~does not change the mean of instrumental distribution but  makes the variance of \\
		~~~~$Ga(a,b)$
		consistent with the asymptotic variance of $\pi(\alpha|\bm{y})$.}
}

\vspace{3mm}
Once the instrumental distribution $Ga(a,b)$ is determined, we simulate $M$ random numbers $\left\{\alpha^{(1)},\alpha^{(2)},\dots,\alpha^{(M)}\right\}$  from $Ga(a,b)$, and
compute the weights $w_i=h_p\left(\alpha^{(i)}\right)/f_{Ga}\left(\alpha^{(i)}|a,b\right)$,
$i=1,\dots,M$, where $f_{Ga}\left(\alpha^{(i)}|a,b\right)$
denotes the PDF value of  $Ga(a,b)$ at $\alpha^{(i)}$. Then normalizing the weights
$\tilde{w}_i=w_i/\sum_{j=1}^{M}w_j$. In the resampling step, we generate
random numbers of $\alpha$ from discrete distribution
\begin{equation}\label{dis2}
	P\left(\alpha=\alpha^{(i)}\right)=\tilde{w}_i,~i=1,2,\dots,M.
\end{equation}
The procedure of posterior inference by SIR is summarized in Algorithm \ref{algo3}.

\begin{algorithm}
	\caption{Sampling importance resampling}\label{algo3}
	\LinesNumbered
	\KwIn{Observed data $\bm{y}$.\\}
	\KwOut{The point estimates and $100(1-\rho)\%$ credible intervals of $\alpha,\beta$ and $\eta$.}
	\BlankLine
	Compute $\delta_p$, $\omega_p$ and $\lambda_p$.
	
	Determine $a$ and $b$ according to the three steps described above.
	
	Generate $M$ random numbers from $Ga(a,b)$, and construct discrete distribution \eqref{dis2}.
	
	\For{$k$ in $\{1,2,\dots,K_3\}$}{
		Generate $\alpha_k$ from discrete distribution \eqref{dis2};
		
		Generate $\beta_k$ from $Ga(\delta_p\bar{T}_m\alpha_{k},\delta_p\lambda_p)$;
		
		Compute $\eta_k=p(\alpha_k,\beta_k)$.
	}	
	Calculate the point and $100(1-\rho)\%$ interval estimates of these parameters by posterior sample $\{(\alpha_k,\beta_k,\eta_k),k=1,\dots,K_3\}$.		
\end{algorithm}

\section{Simulation studies}
\label{sec:sim}
Before performing simulation studies,  a real dataset is analyzed according to the proposed algorithms.
The data are from Meeker and Escobar (1998), which demonstrates the increase in operating current over time for  15 GaAs devices tested at $80^\circ $C.   Measurements of the  increase in operating current are carried out every 250 hours, and the  termination time of the experiment is 4000 hours.
The failure threshold of the device is  10\% increase in the
operating current. Thus, $n=15$, $m=16$ and $\mathbb{C}=10$ in this dataset. The data are shown in Figure \ref{laserdata}, and we can see that the degradation values of  three devices have crossed to the threshold before
test termination time.
Assume that degradation path of the laser device follows gamma process $\mathcal{GP}(\alpha t,\beta)$. Bayesian inference is performed based on  the conjugate prior \eqref{con0}, where $\delta=1$ and
$\lambda=\bar{y}_a$.
$\omega=\bar{y}_g=\prod_{i=1}^{n}\prod_{1}^{m}y_{ij}^{1/(mn)}$ is the geometric mean of $y_{ij}s$, because measurements are evenly spaced.
As we have discussed in Section \ref{sec:prior}, $\delta=1$ means that
the prior information is equivalent to that of taking one measurement.
Compared to data with totally $mn$ measurements, the prior information is quite limited.
Then the posterior distribution of $\alpha$ and $\beta$ is
$AGG\left(mn+1,\prod_{i=1}^{n}\prod_{1}^{m}y_{ij}^{1/(mn)},\bar{y}_a\right)$.
The proposed algorithms are applied to obtain the point estimates and 95\% credible intervals of
$\alpha$ and $\beta$, as well as the reliability of the device at time 4500 hours $R(4500)$.
In the Gibbs sampling, the number of iteration $K_1$ is 3,000 with the first 1,000 burn-in samples and the thinning interval is two. Thus, the effective sample size for posterior inference is 1,000. In the DGS,
the interval for discretization is [0,10], and the number of grid points is 10,000. The sample size for posterior inference is also 1,000. In the SIR, we set $M=10,000$ and $K_3=1,000$.
The results based on the three algorithms are listed in Table \ref{dat1}, where ``GS" denotes the algorithms based on Gibbs sampling. As can be seen in Table \ref{dat1}, the Bayesian point estimates and 95\% credible intervals of $\alpha$, $\beta$ and $R(4500)$ based on the
three algorithms are almost the same.

\begin{figure}[htp]
	\centering
	\includegraphics[width=12cm]{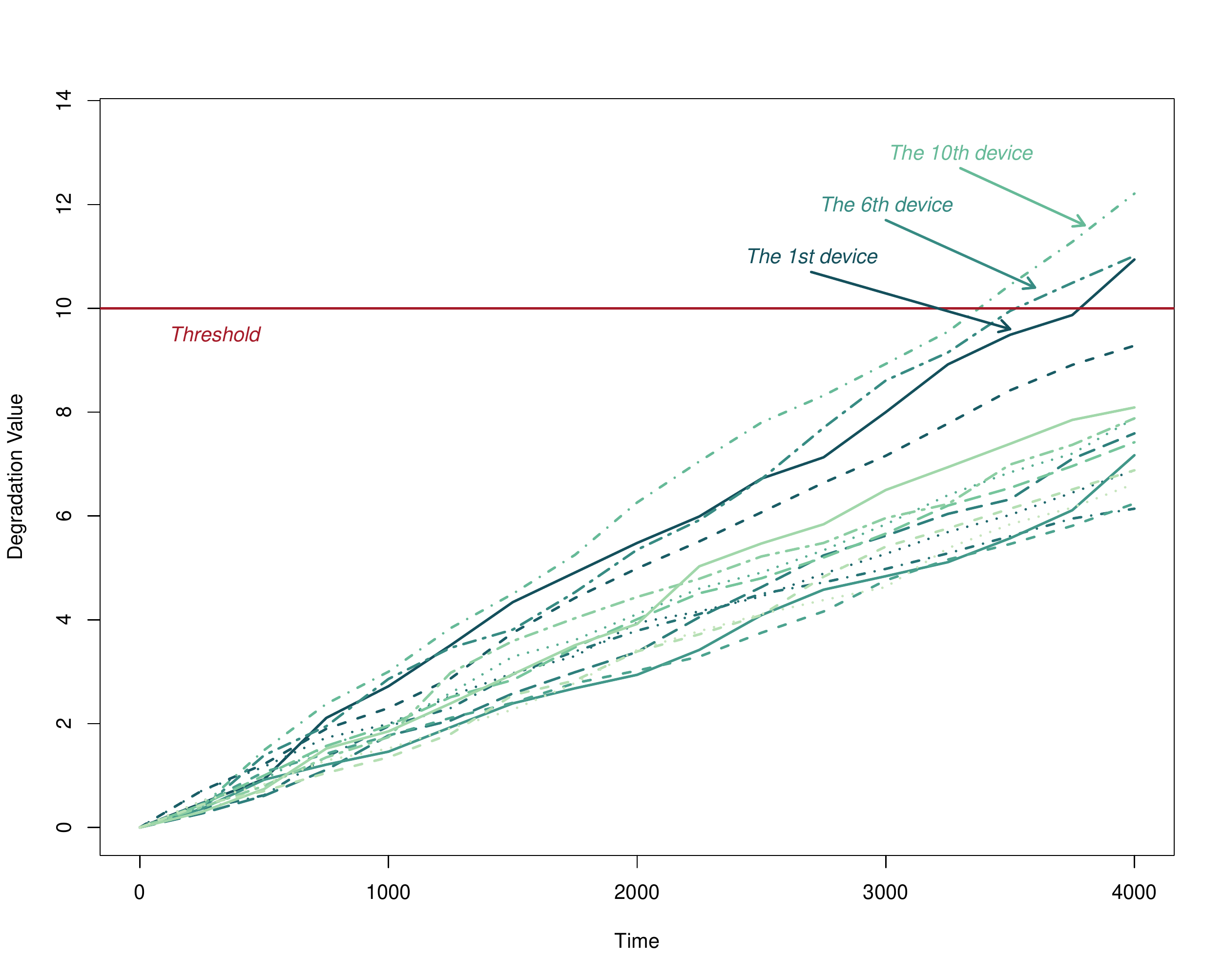}
	\caption{Laser degradation data.}\label{laserdata}
\end{figure}

\begin{table}[htbp]
	\centering
	\caption{The point estimates and 95\% credible intervals of $\alpha$, $\beta$ and $R(4500)$.}\label{dat1}
	\scalebox{1}{	\begin{tabular}{cccccccccc}
			\toprule
			\multirow{2}[4]{*}{Estimate} & \multicolumn{3}{c}{GS} & \multicolumn{3}{c}{DGS} & \multicolumn{3}{c}{SIR} \\
			\cmidrule{2-10}
			&
			$\alpha$	& $\beta$  & $R(4500)$  & $\alpha$	& $\beta$  & $R(4500)$ &$\alpha$	& $\beta$  & $R(4500)$  \\
			\midrule
			Point &  0.0309&15.342&0.879& 0.0308&15.325&0.878&0.0310&15.438&0.882\\
			2.5\%  & 0.0258&12.693& 0.740& 0.0260&12.698&0.737&0.0256&12.677& 0.743\\
			97.5\% & 0.0366&18.332&0.963& 0.0370&18.328&0.962&0.0366&18.368&0.964\\
			\bottomrule
	\end{tabular}}
\end{table}

To compare the three algorithms elaborately, simulation studies are performed under conjugate
priors with varying amounts of information. The data are generated from $\mathcal{GP}(\alpha t,\beta)$,
where $\alpha=0.031$ and $\beta=15.35$ (close to the estimates in Table \ref{dat1}).
There are a total of $n=15$ specimens tested, and each specimen is  measured every 250 hours.
The measurement times $m=16$. The failure threshold is set as 10. The conjugate prior is
$AGG(\delta,\bar{y}_g, \bar{y}_a)$, and we choose $\delta=0,1,\frac{m}{4},\frac{m}{2}$ to
evaluate the impact of the prior information content on the results.

We generate $N=10,000$ datasets, and the proposed algorithms are used to obtain the Bayesian point estimates
and 95\% credible intervals
of $\alpha$, $\beta$, $R(4500)$ and MTTF for each dataset. Then based on the 10,000 Bayesian point estimates,  absolute relative error (RB) and  root mean square error (RMSE) are computed, that is,
$$\text{RB}=\dfrac{1}{N}\sum_{i=1}^{N}\Bigg|\dfrac{\text{Estimate}_i-\text{True value}}{\text{True value}}\Bigg|,~
\text{RMSE}=\sqrt{\dfrac{1}{N}\sum_{i=1}^{N}\left(\text{Estimate}_i-\text{True value}\right)^2}.$$
The results are reported in Tables \ref{arb} and \ref{rmse}. Overall,  the parameters can be satisfactorily
estimated in all the scenarios, as the RBs of Bayesian estimates of $\alpha$, $\beta$ and $R(4500)$ are around
2\%, and the RBs of Bayesian estimates of MTTF are around 0.1\%. For both RB and RMSE, the three algorithms
perform nearly identically, and the influence of $\delta$ (different amount of prior information) on the estimates
are insignificant.
For interval estimate  of the parameters, we calculate average length and
frequentist coverage probability (FCP) of the 95\% credible intervals. The results are listed in Tables
\ref{length} and \ref{fcp}.  An interesting phenomenon lies in the lengths of 95\% credible intervals.
The intervals narrow as the amount of prior information increases, and the difference among three algorithms is insignificant.
However, the FCPs display a different pattern. For the model parameters $\alpha$ and $\beta$,
the FCPs based on DGS and SIR are much closer to the nominal level than these based on GS
regardless of $\delta$.
While for $R(4500)$ and MTTF,  the FCPs based on the three algorithms
are  always very close to the nominal level, indicating the superiority of the proposed posterior sampling algorithms.
The almost accurate quantification of the uncertainties evidently provides precise information on the system reliability and MTTF, which is useful in RUL prediction.

\begin{table}[htbp]
	\centering
	\caption{RBs of point estimates of the parameters.}\label{arb}
	\scalebox{1}{	\begin{tabular}{ccccccccc}
			\toprule
			\multirow{2}[4]{*}{Algorithm} & \multicolumn{4}{c}{$\delta=0$} & \multicolumn{4}{c}{$\delta=1$} \\
			\cmidrule{2-9}
			&
			$\alpha$	& $\beta$  & $R(4500)$ & MTTF & $\alpha$	& $\beta$  & $R(4500)$ & MTTF \\
			\midrule
			GS &  0.0245&	0.0256	& 0.0161&  0.00109&	0.0243&	0.0254&	 0.0161&	0.00108\\
			DGS  & 0.0245&	0.0256&	 0.0161&	0.0011&	0.0245&	0.0256&	0.0161&	0.00109\\
			SIR & 0.0245&	0.0256&	0.0161&	0.00109&	0.0245&	0.0256&	0.0161&	0.00109\\
			\midrule
			\multirow{2}[4]{*}{Algorithm} & \multicolumn{4}{c}{$\delta=\frac{m}{4}$} & \multicolumn{4}{c}{$\delta=\frac{m}{2}$} \\
			\cmidrule{2-9}
			&
			$\alpha$	& $\beta$  & $R(4500)$ & MTTF & $\alpha$	& $\beta$  & $R(4500)$ & MTTF \\
			\midrule
			GS &  0.0233&	0.0247&	0.0153&	0.00136&	0.0234&	0.0248&	0.0151&	0.00137\\
			DGS  & 0.0234&	0.0249&	0.0152&	0.00137&	0.0233&	0.0247&	0.0151&	0.00136\\
			SIR & 0.0234&	0.0249&	0.0152&	0.00138&	0.0232&	0.0246&	0.0151&	0.00136\\
			\bottomrule
	\end{tabular}}
\end{table}

\begin{table}[htbp]
	\centering
	\caption{RMSEs of point estimates of the parameters.}\label{rmse}
	\scalebox{1}{	\begin{tabular}{ccccccccc}
			\toprule
			\multirow{2}[4]{*}{Algorithm} & \multicolumn{4}{c}{$\delta=0$} & \multicolumn{4}{c}{$\delta=1$} \\
			\cmidrule{2-9}
			&
			$\alpha$	& $\beta$  & $R(4500)$ & MTTF & $\alpha$	& $\beta$  & $R(4500)$ & MTTF \\
			\midrule
			GS &  0.00302&	1.547&	0.0601&	115.218&	0.00302&	1.547&	0.0601&	115.205\\
			DGS  & 0.00301&1.539&0.0601&	115.267&	0.00301&	1.538&	0.0601&	115.258\\
			SIR & 0.00301&	1.539&	0.0601&	115.267&	0.00301&	1.538&	0.0601&	115.258\\
			\midrule
			\multirow{2}[4]{*}{Algorithm} & \multicolumn{4}{c}{$\delta=\frac{m}{4}$} & \multicolumn{4}{c}{$\delta=\frac{m}{2}$} \\
			\cmidrule{2-9}
			&
			$\alpha$	& $\beta$  & $R(4500)$ & MTTF & $\alpha$	& $\beta$  & $R(4500)$ & MTTF \\
			\midrule
			GS &  0.00302&1.546&	0.0599&	115.204&	0.00296&	1.512&	0.0607&	116.192\\
			DGS  & 0.00301& 1.537&0.0600&115.250&	0.00294&1.501&	0.0607&	116.205\\
			SIR & 0.00301&	1.537&	0.0600&	115.250&	0.00294&	1.501&	0.0607&	116.205\\
			\bottomrule
	\end{tabular}}
\end{table}

\begin{table}[htbp]
	\centering
	\caption{Lengths of 95\% credible intervals of the parameters.}\label{length}
	\scalebox{1}{	\begin{tabular}{ccccccccc}
			\toprule
			\multirow{2}[4]{*}{Algorithm} & \multicolumn{4}{c}{$\delta=0$} & \multicolumn{4}{c}{$\delta=1$} \\
			\cmidrule{2-9}
			&
			$\alpha$	& $\beta$  & $R(4500)$ & MTTF & $\alpha$	& $\beta$  & $R(4500)$ & MTTF \\
			\midrule
			GS &  0.0109&	5.588&	0.224&	447.444&	0.0109&	5.585&	0.224&	446.609\\
			DGS  & 0.0109&	5.629&	0.224&	446.274	&0.0109&	5.620&	0.223&	445.596\\
			SIR & 0.0110&	5.630&	0.224&	446.414&	0.0109&	5.624&	0.223&	445.318\\
			\midrule
			\multirow{2}[4]{*}{Algorithm} & \multicolumn{4}{c}{$\delta=\frac{m}{4}$} & \multicolumn{4}{c}{$\delta=\frac{m}{2}$} \\
			\cmidrule{2-9}
			&
			$\alpha$	& $\beta$  & $R(4500)$ & MTTF & $\alpha$	& $\beta$  & $R(4500)$ & MTTF \\
			\midrule
			GS &  0.0108&	5.541&	0.222&	444.082&	0.0107&	5.502&	0.220&	440.493\\
			DGS  & 0.0108&	5.582&	0.221&	443.013	&0.0107&	5.535&	0.219&	439.562\\
			SIR & 0.0109&	5.581&	0.221&	443.131&	0.0108&	5.536&	0.219&	439.373\\
			\bottomrule
	\end{tabular}}
\end{table}

\begin{table}[htbp]
	\centering
	\caption{Frequentist coverage probabilities of 95\% credible intervals of the parameters.}\label{fcp}
	\scalebox{1}{	\begin{tabular}{ccccccccc}
			\toprule
			\multirow{2}[4]{*}{Algorithm} & \multicolumn{4}{c}{$\delta=0$} & \multicolumn{4}{c}{$\delta=1$} \\
			\cmidrule{2-9}
			&
			$\alpha$	& $\beta$  & $R(4500)$ & MTTF & $\alpha$	& $\beta$  & $R(4500)$ & MTTF \\
			\midrule
			GS &  0.9372&	0.9369&	0.9474&	0.9492&	0.9344&	0.9336&	0.9458&	0.9475\\
			DGS  & 0.9424&	0.9424&	0.9463&	0.9481&	0.9416&	0.9396&	0.9460&	0.9468\\
			SIR & 0.9418&	0.9406&	0.9467&	0.9473&	0.9411&	0.9413&	0.9471&	0.9484\\
			\midrule
			\multirow{2}[4]{*}{Algorithm} & \multicolumn{4}{c}{$\delta=\frac{m}{4}$} & \multicolumn{4}{c}{$\delta=\frac{m}{2}$} \\
			\cmidrule{2-9}
			&
			$\alpha$	& $\beta$  & $R(4500)$ & MTTF & $\alpha$	& $\beta$  & $R(4500)$ & MTTF \\
			\midrule
			GS &  0.9340&	0.9337&	0.9455&	0.9478&	0.9385&	0.9369&	0.9399&	0.9413\\
			DGS  & 0.9394&	0.9384&	0.9454&	0.9466&	0.9434&	0.9428&	0.9411&	0.9423\\
			SIR & 0.9406&	0.9384&	0.9451&	0.9456&	0.9451&	0.9424&	0.9415&	0.9427\\
			\bottomrule
	\end{tabular}}
\end{table}

The average computational time of the three algorithms for each dataset
is 0.602, 0.00341 and 0.00499 seconds in
a desktop with Intel(R) Core(TM) i7-10700 CPU at 2.9 GHz and 16 GB RAM running under a Windows 11 operating system. Therefore, the computational efficiency of the DGS and SIR algorithms is comparable, which
are more than one hundred times faster than the GS algorithm. Computational efficiency
is an important index  in the scenario of
online inference, because the posterior distribution is updated in real-time as
new observations are collected and
posterior inference needs to be completed as soon as possible on the premise of ensuring the estimation accuracy. As listed in Tables \ref{arb}-\ref{fcp},
DGS and SIR algorithms are not only high
efficient in terms of computation, but also have high estimation accuracy, which meets the requirements of  online inference. In the following sections,  we mainly utilize the two algorithms
to predict RUL online.

%		However, DGS algorithm needs to specify an interval for discretization in advance, which also requires some exploration time. Determination of interval may also be a drawback of DGS algorithm in the scenario of
%		online inference, because the posterior distribution will be updated in real time after
%		new observations collected, which causes interval for discretization to change constantly.
%		While the implementation of GS and SIR algorithms does not contain any predetermined quantities,
%		and they are more flexible than DGS algorithm. 	Besides, as listed in Table \ref{fcp}, the FCPs of $\alpha$
%		based on DGS algorithm are much farther to the nominal level than these based on the other two algorithms.
%		In general, the efficiency of DGS algorithm is lower than that of
%		GS and SIR algorithms.
%		

\section{Heterogeneity}
\label{sec:heter}
Heterogeneity usually exists among systems because of endogenous and  exogenous factors.
Endogenous factors could include variations in raw materials and assembly lines, while the exogenous factors could be
due to variations in operating environments and usage patterns.
Heterogeneity will cause the performance degradation of each system to show different patterns.
However, the systems come from the same population, and their failure mechanisms are consistent.
Thus, we assume that degradation  of the $i$-th system's PC follows gamma process $\mathcal{GP}(\alpha t,\beta_i)$ in this section. The  same shape parameter $\alpha$ denotes a common failure mechanism among systems, and
different scales $\beta_i$s represent heterogeneity existed among systems.

For the sake of simplifying notations, we assume that the measurements are equally spaced. That is, the lag between
two adjacent measurement time epochs is $l$. Assume that there are $n$ systems tested in the experiment.
Until time epoch $T_m=ml$, we have measured the degradation values of all the $n$ systems $m$ times.
Let $Y_{ij}$ be the degradation value of the $i$-th system at time epoch $T _j=jl$, $i=1,\dots,n,$ $j=1,\dots,m$.
The degradation increment $y_{ij}=Y_{ij}-Y_{ij-1}$, where $Y_{i0}=0$, $i=1,\dots,n,$ $j=1,\dots,m$.
At time epoch $T_m$, the observed data are $\bm{y_{(m)}}=\{y_{ij}, i=1,\dots,n,~j=1,\dots,m\}$.
Since $\mathcal{Y}_i(t)\sim \mathcal{GP}(\alpha t,\beta_i)$, we have $y_{ij}\sim Ga(\alpha l,\beta_i)$.
Then based on $\bm{y_{(m)}}$,
the likelihood function is
\begin{equation}\label{like2}
	\begin{aligned}
		L\left( \bm{y_{(m)}}|\alpha,\beta_1,\dots,\beta_n\right)&=\prod_{i=1}^{n}\prod_{j=1}^{m}\dfrac{\beta_i^{\alpha l}}{\Gamma(\alpha l)}
		y_{ij}^{\alpha l-1}\exp\{-\beta_i y_{ij}\}\\
		&\propto \dfrac{\bar{\beta}_g^{mnl\alpha}}{\left[\Gamma(\alpha l)\right]^{mn}}
		\bar{y}_{g(m)}^{mnl\alpha}\exp\left\{-\sum_{i=1}^{n}m\bar{y}_{i(m)}\beta\right\},
	\end{aligned}
\end{equation}
where $\bar{\beta}_{g}=\left[\prod_{i=1}^{n}\beta_{i}\right]^{1/n}$,
$\bar{y}_{g(m)}=\left[\prod_{i=1}^{n}\prod_{j=1}^{m}y_{ij}\right]^{\frac{1}{mn}}$ and
$\bar{y}_{i(m)}=\frac{1}{m}\sum_{j=1}^{m}y_{ij}$, $i=1, \dots, n$.

\begin{thy}
	Given the likelihood function \eqref{like2}, a conjugate prior for $(\alpha,\beta_1,\dots,\beta_n)^{'}$ is
	\begin{equation}\label{con2}
		\begin{aligned}
			\pi(\alpha,\beta_1,\dots,\beta_n)&=
			C \dfrac{\left(\bar{\beta}_g\omega\right)^{\delta_1 l\alpha}}{\left[\Gamma(l\alpha)\right]^{\delta_1}}
			\exp\left\{-\sum_{i=1}^{n}\delta_2\lambda_i\beta_i\right\},
		\end{aligned}
	\end{equation}
	where $C$ is a normalized constant, $\delta_1$, $\delta_2$, $\omega$ and $\lambda_i$s are hyperparameters with nonnegative values.
\end{thy}
\begin{proof}[Proof:]
	Based on the likelihood function \eqref{like2} and the prior \eqref{con2}, the joint posterior
	density of $(\alpha,\beta_1,\dots,\beta_n)^{'}$ is
	\begin{equation}\label{post2}
		\begin{aligned}
			\pi(\alpha,\beta_1,\dots,\beta_n|\bm{y})&\propto 	L( \bm{y}|\alpha,\beta_1,\dots,\beta_n)\pi(\alpha,\beta_1,\dots,\beta_n)\\
			& \propto \dfrac{\bar{\beta}_{g}^{(mn+\delta_1)l\alpha}\bar{y}_{g(m)}^{mnl\alpha}\omega^{\delta_1 l\alpha}}{\left[\Gamma(l\alpha)\right]^{mn+\delta_1}}
			\exp\left\{-\sum_{i=1}^{n}\left(m\bar{y}_{i(m)}+\delta_2\lambda_i\right)\beta_i\right\}\\
			&\propto \dfrac{\left(\bar{\beta}_{g}\omega_{p(m)}\right)^{\delta_{1p(m)}l\alpha}}{\left[\Gamma(l\alpha)\right]^{\delta_{1p(m)}}}
			\exp\left\{-\sum_{i=1}^{n}\delta_{2p(m)}\lambda_{ip(m)}\beta_i\right\},
		\end{aligned}
	\end{equation}
	where $\delta_{1p(m)}=mn+\delta_1$, $\delta_{2p(m)}=m+\delta_2$,
	$\omega_{p(m)}=\omega^{\frac{\delta_1}{\delta_{1p(m)}}}\bar{y}_{g(m)}^{\frac{mn}{\delta_{1p(m)}}}$ and $\lambda_{ip(m)}=\frac{m}{\delta_{2p(m)}}\bar{y}_{i(m)}+\frac{\delta_2}{\delta_{2p(m)}}\lambda_i$, $i=1,\dots,n$. Thus, $\pi(\alpha,\beta_1,\dots,\beta_n)$
	and $\pi(\alpha,\beta_1,\dots,\beta_n|\bm{y})$ are from the same distribution family.
\end{proof}
When $\beta_1=\dots=\beta_n=\beta$, $\delta_1=n\delta_2=\delta$ and $\lambda_1=\dots=\lambda_n=\lambda$,
the conjugate prior \eqref{con2} is reduced to \eqref{con1}.
For bettering understanding the conjugate prior \eqref{con2}, we rewrite $\pi(\alpha,\beta_1,\dots,\beta_n)$ as
\begin{equation*}
	\begin{aligned}
		\pi(\alpha,\beta_1,\dots,\beta_n)&=\prod_{i=1}^{n}\pi(\beta_i|\alpha)\pi(\alpha)\\
		&\propto \prod_{i=1}^{n}\dfrac{(\delta_2\lambda_i)^{1+\delta_1l\alpha/n}\beta_i^{\delta_1 l\alpha/n}}{\Gamma(1+\delta_1 l\alpha/n)}\exp\{-\delta_2\lambda_i\beta_i\}\\
		&~\times \dfrac{\left[\Gamma(1+\delta_1 l\alpha/n)\right]^n}{\left[\Gamma(l\alpha)\right]^{\delta_1}}
		\exp\left\{-\alpha\delta_1 l\left[\log\left(\frac{\delta_2}{\omega}\right)+
		\frac{1}{n}\sum_{i=1}^{n}\log\lambda_i\right]\right\}.	
	\end{aligned}
\end{equation*}
Thus, given $\alpha$, the conditional density of  $\beta_i$ is $Ga(1+\delta_1 l\alpha/n,\delta_2\lambda_i)$,
and  the marginal density of $\alpha$ is proportional to
\begin{equation}\label{margi2}
	g(\alpha)=\dfrac{\left[\Gamma(1+\delta_1 l\alpha/n)\right]^n}{\left[\Gamma(l\alpha)\right]^{\delta_1}}	\exp\left\{-\alpha\delta_1 l\left[\log\left(\frac{\delta_2}{\omega}\right)+
	\frac{1}{n}\sum_{i=1}^{n}\log\lambda_i\right]\right\}.	
\end{equation}
Using Stirling's formula and as $\alpha\rightarrow\infty$, we have
$$g(\alpha)\equiv O\left(\alpha^{\frac{\delta_1+n}{2}}
\exp\left\{-A\alpha\right\}\right),$$
where $A=\delta_1 l\left[\log\left(\frac{n\delta_2}{\delta_1}\right)+\frac{1}{n}\sum_{i=1}^{n}\log\left(\frac{\lambda_i}{\omega}\right)\right]$.
Then the tail of $\pi(\alpha)$ can be approximated by $Ga\left(\frac{\delta_1+n+2}{2},K\right)$ when $A>0$.
Due to the tail property of $\pi(\alpha)$,
we call $\pi(\alpha,\beta_1,\dots,\beta_n)$ approximated-gamma-multivariate-gamma (AGMG) distribution with dimension $n$, denoted as $AGMG_n(\bm\gamma,\omega,\bm\xi)$, where $\bm\gamma=(\delta_1,\delta_2)^{'}$, and
$\bm\xi=(\lambda_1,\dots,\lambda_n)^{'}$.

Based on \eqref{post2}, we know that the posterior of $(\alpha,\beta_1,\dots,\beta_n)^{'}$ is
$AGMG_n\left(\bm\gamma_{p(m)},\omega_{p(m)},\bm\xi_{p(m)}\right)$, where $\bm\gamma_{p(m)}=(\delta_{1p(m)},\delta_{2p(m)})^{'}$
and $\bm\xi_{p(m)}=(\lambda_{1p(m)},\dots,\lambda_{np(m)})^{'}$. Special choices for hyperparameters $\omega$
and $\bm\xi$ are
$\bar{y}_{g(m)}$ and $\bm{\bar{y}}_{(m)}=\left(\bar{y}_{1(m)},\dots,\bar{y}_{n(m)}\right)^{'}$, respectively.
In this setting, the hyperparameters $\delta_1$ and  $\delta_2$ behave like number of measurements, because the posterior will be $AGMG_n(\bm\gamma_{p(m)},\bar{y}_{g(m)},\bm{\bar{y}}_{(m)})$. Similar to \eqref{con1},
$\delta_1$ and  $\delta_2$ mainly describe kurtosis of AGMG distribution, which control the belief of prior
information. The generation of random numbers from AGMG distribution can be implemented by the algorithms
\ref{algo2} and \ref{algo3}
with slight modifications. The main difference is that the optimization object is replaced by the posterior marginal distribution of $\alpha$ in $AGMG_n\left(\bm\gamma_{p(m)},\omega_{p(m)},\bm\xi_{p(m)}\right)$, and given $\alpha$,
$\beta_i$ is simulated from gamma distribution $Ga\left(1+\delta_{1p(m)} l\alpha/n,\delta_{2p(m)}\lambda_{ip(m)}\right)$,
$i=1,\dots,n$. The computational time of the two algorithms (DGS and SIR) is proportional to the dimension $n$. As an illustration,
we implement the two algorithms for AGMG distributions with  $n$ from 2 to 50 under the same parameter settings,
and the computational time of the two algorithms is shown in
Figure \ref{ctime}. As can be seen in Figure \ref{ctime},
the computational time grows linearly as $n$. When $n$ increases from 2 to 50, the computational time required by DGS  increases from 0.00328 seconds to 0.00895 seconds, and for SIR,
it increases  from 0.00448 seconds to 0.0101 seconds. Therefore, both algorithms
have high computational efficiency, even for large $n$.

\begin{figure}[htp]
	\centering
	\includegraphics[width=13cm]{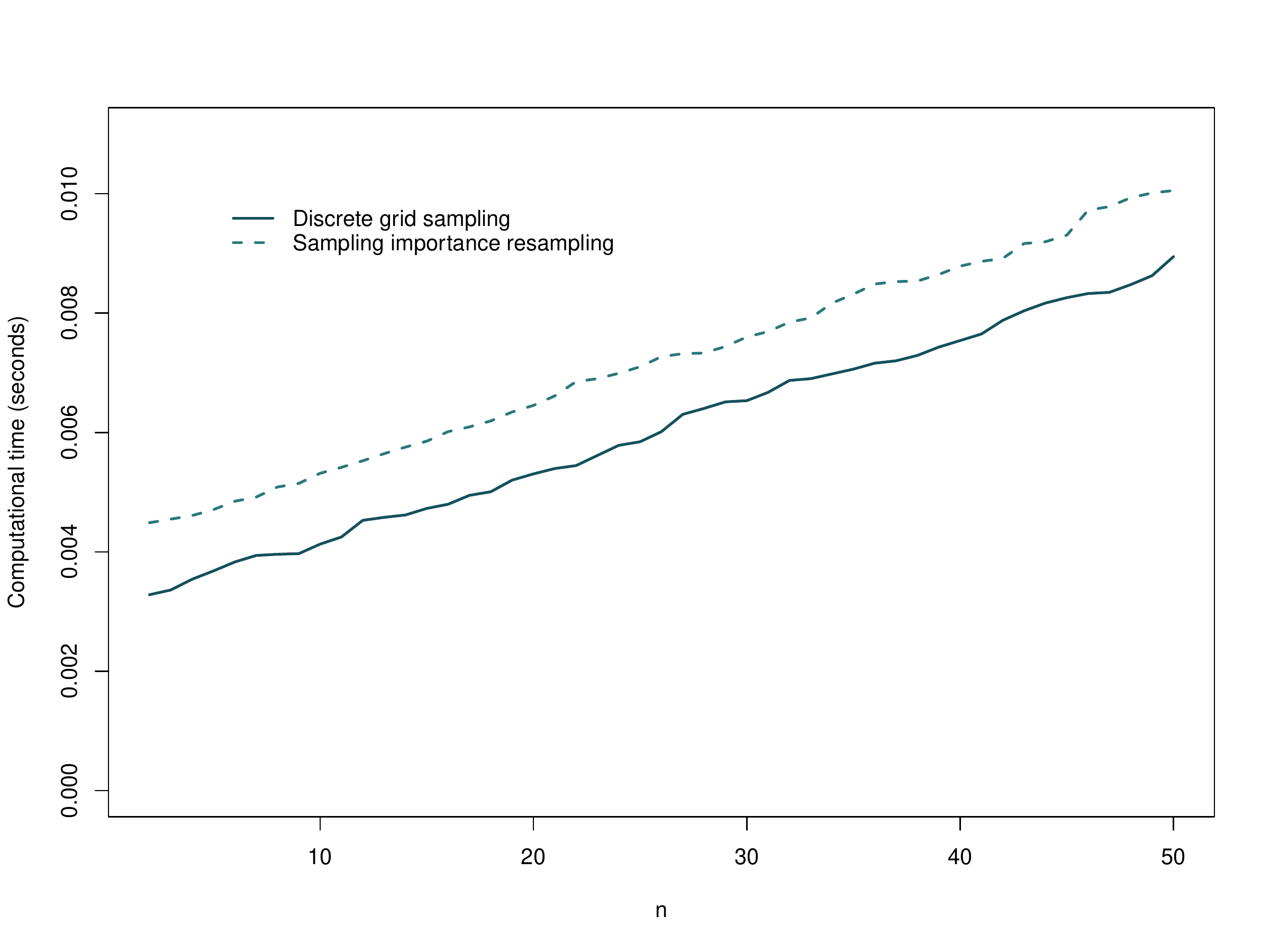}
	\caption{Computational time of two algorithms for sampling AGMG distribution with different dimensions.}\label{ctime}
\end{figure}

For a functional system, one may be interested in the indexes related to RUL, which also serves as the foundation for
prognostics and health management.  Assume that all the degradation values of the $i$-th system until time $t_m$ are	less than $\mathcal{C}$. The RUL of the  $i$-th system at time $t_m$ is defined as
$$Z_{it_m}=\inf\{z: \mathcal{Y}_i(z+t_m)\ge\mathbb{C}|Y_{im}<\mathbb{C}\}.$$
The CDF of  $Z_{i t_m}$ is
\begin{equation}
	\begin{aligned}
		F_{Z_{it_m}}(z|\alpha,\beta_i)&=P(Z_{it_m}< z)=P(\mathcal{Y}_i(z+t_m)>\mathbb{C})\\
		&=P(\mathcal{Y}_i(z)>\mathbb{C}-Y_{im})
		=\frac{\Psi(\beta(\mathbb{C}-Y_{im}),\alpha z)}{\Gamma(\alpha z)},
	\end{aligned}
\end{equation}
where the last two equalities hold because of the homogeneous property of the gamma process.
Due to the complicated form of $F_{Z_{it_m}}(z|\alpha,\beta_i)$,
we use a two-parameter Birnbaum-Saunders distribution to approximate
$F_{Z_{it_m}}(z|\alpha,\beta_i)$, which can greatly simplify the function form.
According to \cite{parpad2005},
$BS(\alpha_{im}^\ast,\beta_{im}^\ast)$ with CDF $\Phi\left(\frac{1}{\alpha_{im}^\ast}\left[\sqrt{\frac{z}{\beta_{im}^\ast}}-\sqrt{\frac{\beta_{im}^\ast}{z}}\right]\right)$
has a high approximation precision, where $\alpha_{im}^\ast=\sqrt{\frac{1}{\beta_i(\mathbb{C}-Y_{im})}}$ and
$\beta_{im}^\ast=\frac{\beta_i(\mathbb{C}-Y_{im})}{\alpha}$.
Then mean of $Z_{i t_m}$ can be approximated by
$$\mu_{im}(\alpha,\beta_i)=\beta_{im}^\ast\left(1+\left(\alpha_{im}^\ast\right)^2/2\right)=
\dfrac{1+2\beta_i(\mathcal{C}-Y_{im})}{2\alpha}.$$
The lower $\rho$-th quantile of the distribution of  $Z_{i t_m}$ can be approximated by
$$\mu^{\rho}_{im}(\alpha,\beta_i)=\frac{\beta_{im}^\ast}{4}\left[u_\rho\alpha_{im}^\ast +\sqrt{\left(u_\rho\alpha_{im}^\ast\right)^{2}+4}\right]^{2},$$
where $u_\rho$ is the $\rho$-th quantile of the standard normal distribution.
Bayesian point prediction of $RUL$ of the $i$-th system at time $t_m$:
\begin{equation}\label{point}
	\tilde{\mu}_{im}=\int_0^{\infty}\int_0^{\infty}\mu_{im}(\alpha,\beta_i)\pi(\alpha,\beta_i| \bm{y_{(m)}})\text{d}\alpha\text{d}\beta_i.
\end{equation}
Bayesian interval prediction of $RUL$ of the $i$-th system at time $t_m$ with $1-\rho$ credible level:
\begin{equation}\label{inter}
	\left(\tilde{\mu}_{im}^{\rho/2},\tilde{\mu}_{im}^{1-\rho/2}\right),
\end{equation}
where $\tilde{\mu}_{im}^{\rho}=\int_0^{\infty}\int_0^{\infty}\mu^{\rho}_{im}(\alpha,\beta_i)\pi(\alpha,\beta_i| \bm{y_{(m)}})\text{d}\alpha\text{d}\beta_i$.  Given the posterior samples
$\{(\alpha^{(k)},\beta_i^{(k)}), k=1,\dots,K\}$, \eqref{point} and
\eqref{inter} can be approximated by Monte Carlo methods:
\begin{equation}\label{mcest}
	\tilde{\mu}_{im}\approx\frac{1}{K}\sum_{k=1}^{K}\mu_{im}(\alpha^{(k)},\beta_i^{(k)}),~~
	\tilde{\mu}_{im}^{\rho}\approx\frac{1}{K}\sum_{k=1}^{K}\mu^{\rho}_{im}(\alpha^{(k)},\beta_i^{(k)}).	
\end{equation}

{\bf Remark:} DGS and SIR are proposed to produce posterior samples of
the model parameters $(\alpha,\beta_1,\dots,\beta_n)^{'}$. Based on the posterior samples, the RUL prediction for $n$ systems can be carried out by \eqref{mcest} simultaneously. The algorithms are flexible and can be used for single or multiple systems. When $n=1$, RUL is learned by the information from a single system, and the algorithms are reduced to these in Section \ref{sec:ps}.
When $n\ge 2$, this is a strategy for collaborative learning 
that makes use of the full information from multiple systems to estimate the common parameter $\alpha$, and the posteriors of heterogeneous parameters $\beta_i$s indirectly borrow the information from other systems to assist in improving the estimation accuracy.

\section{Online RUL prediction}
\label{sec:rul}
With the rise and popularization of advanced sensor technology, the performance degradation information of the system can be monitored in real-time, and RUL prediction will be updated after new observations are collected.
In this section, we will propose an online RUL prediction algorithm based on the gamma process. The proposed algorithm possesses several advantages for online updating, such as
high computational efficiency,  low requirement for data storage space, RUL prediction for multiple systems simultaneously, etc.

Assume that new degradation increments $(y_{1m+1},\dots,y_{nm+1})$ of $n$ systems are collected at time $t_{m+1}=(m+1)l$. Then the posterior distribution of $(\alpha,\beta_1,\dots,\beta_n)^{'}$ needs to be updated after new observations arriving.
For Bayesian inference with conjugate priors, 
a recursive formula can be used to implement the updating. From \eqref{post2},
we know that the posterior distribution of $(\alpha,\beta_1,\dots,\beta_n)^{'}$ at time
$t_{m+1}=(m+1)l$ is $AGMG_n\left(\bm\gamma_{p(m+1)},\omega_{p(m+1)},\bm\xi_{p(m+1)}\right)$,
where the parameters $\bm\gamma_{p(m+1)}$, $\omega_{p(m+1)}$ and $\bm\xi_{p(m+1)}$
can be updated recursively. That is,
\begin{equation}\label{post3}
	\begin{aligned}
		\bm\gamma_{p(m+1)}&=\bm\gamma_{p(m)}+(n,1)^{'}, \\
		\omega_{(m+1)}&= \omega_{(m)}^{\frac{mn+\delta_1}{(m+1)n+\delta_1}}\left[\prod_{i=1}^{m}y_{im+1}\right]^{\frac{1}{(m+1)n+\delta_1}},\\
		\bm\lambda_{(m+1)}&=\frac{m+\delta_2}{m+1+\delta_2}\bm\lambda_{(m)}+
		\frac{1}{m+1+\delta_2}\left(y_{1m+1},\dots,y_{nm+1}\right)^{'}.	
	\end{aligned}
\end{equation}
The recursive formula for posterior distribution can significantly reduce
data storage space, since only the values of parameters in posterior distribution and
new observations need to be recorded in \eqref{post3}. Besides, we have proposed two algorithms
with high computational efficiency to obtain estimates for AGMG distribution.
Thus, the online RUL prediction for multiple systems can be summarized as follows.

\begin{algorithm}
	\caption{Online RUL prediction}\label{algo4}
	\LinesNumbered
	\KwIn{Parameter values in posterior distribution of $(\alpha,\beta_1,\dots,\beta_n)^{'}$ at time $t_{m}=ml$, and new observations $(y_{1m+1},\dots,y_{nm+1})$ at time $t_{m+1}=(m+1)l$.\\}
	\KwOut{The point estimates and $100(1-\rho)\%$ credible intervals of RULs of $n$ systems.}
	\BlankLine
	Update  posterior distribution of $(\alpha,\beta_1,\dots,\beta_n)^{'}$ by \eqref{post3}.
	
	Generate random numbers from the updated posterior distribution by DGS or SIR.
	
	Calculate the point and $100(1-\rho)\%$ interval estimates of RULs of $n$ systems
	by 	\eqref{mcest}.
\end{algorithm}

%Algorithm \ref{algo4}

\section{Case study}
\label{sec:case}
\subsection{Laser degradation data}
The laser degradation data taken into account in Section \ref{sec:sim} have been reanalyzed.
As shown in Figure \ref{dat1}, the first, sixth and tenth laser devices have failed
at 4,000 hours, because the degradation values crossed the threshold level 10. The exact failure time of the three devices are unknown. However, we know that failure time lies in certain time intervals. For instance, the degradation value of the first device
crosses the threshold between 3750 and 4000 hours. Thus, the linear interpolation method can be used to estimate the failure time. The degradation values of the first device at time epochs
3750 and 4000 hours are 9.87 and 10.94, respectively. Using the linear interpolation method,
its failure time can be estimated by
$$3750+\dfrac{10-9.87}{10.94-9.87}\times (4000-3750)=3785.75~~ \text{hours}.$$
Similarly, the failure time of the sixth and tenth devices are
3506.75 and 3351.25 hours, respectively.

To illustrate the online algorithm, we start to predict the RULs of the three devices
at 500 hours, which means that  only two measurements  are utilized at the beginning.
Then the posterior distribution of the parameters is updated when new measurements
are involved. Algorithm \ref{algo4} is used to obtain the
estimates of the parameters,
the point and $95\%$ interval estimates of RULs of the three devices at each time epoch.
Figure \ref{paresti} shows the estimates of $\alpha$, $\beta_1$, $\beta_6$ and $\beta_{10}$ at each time epoch. From Figure \ref{paresti}, we can see that
the estimates of $\beta_1$, $\beta_6$ and $\beta_{10}$ have an increasing trend,
and their values have a significant difference. This implies that the devices are heterogeneous.
%While
%the estimates of $\alpha$ are relative stable after 1,750 hours, which lies in
%(0.022,0.030).
The point predictions and 95\% predictive intervals of RULs of the three devices are shown in
Figure \ref{rulpred}, in which we also display the true RULs of the three devices
at each time epoch. As can seen in Figure \ref{rulpred}, almost all the true RULs are covered by 95\% credible intervals. Furthermore, the point predictions are extremely close to the true RULs.

\begin{figure}[htp]
	\centering
	\includegraphics[width=13cm]{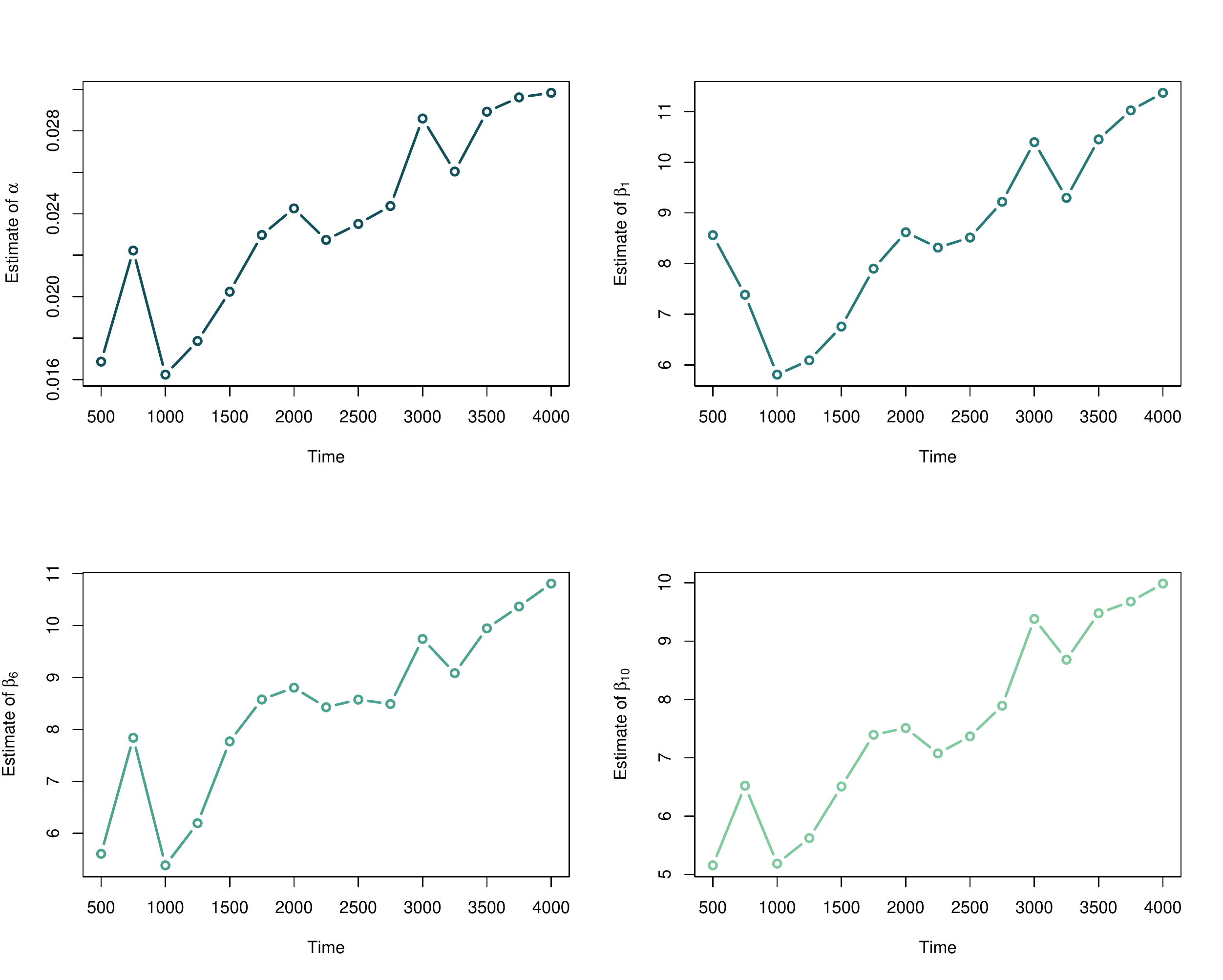}
	\caption{Online estimates of $\alpha$, $\beta_1$, $\beta_6$ and $\beta_{10}$ for laser degradation data.}\label{paresti}
\end{figure}

\begin{figure}[htp]
	\centering
	\includegraphics[width=15cm]{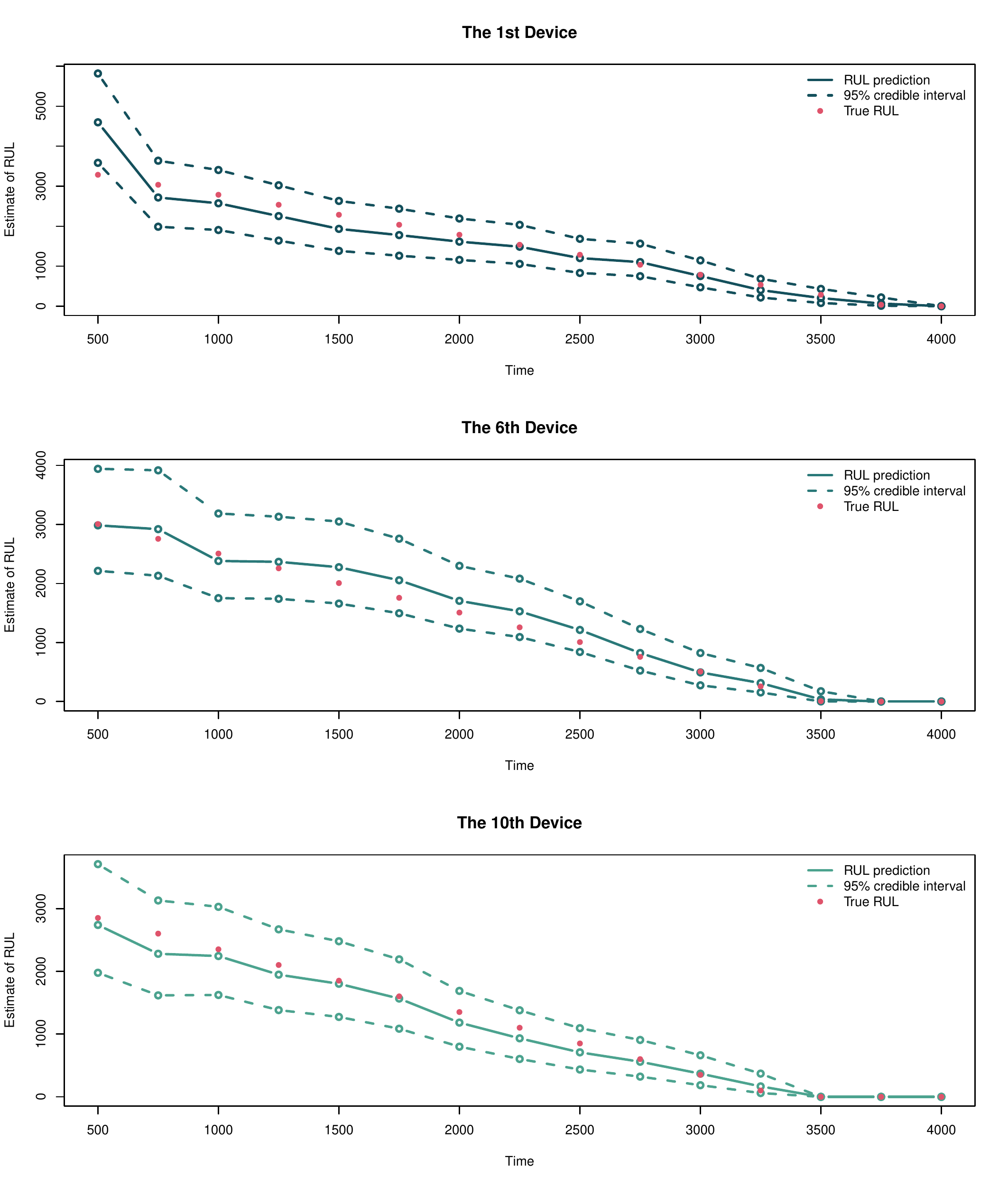}
	\caption{The point predictions and 95\% predictive intervals of RULs for the first, sixth and tenth devices.}\label{rulpred}
\end{figure}

\subsection{Train wheel data}
The train wheel data is from Almeida (2011, Table 5.1, p. 69).
The wheels will wear down with distance driven. When the wear of  wheel diameter
attains 60 (mm), the wheel is considered to have failed.
The original data set includes
14  specimens and the measurements are implemented equally spaced
from 50 to 600  in increments of 50, where the  unit of distance is 1000 km (kkm).
The main goal of this section is to predict the RULs of wheels after each measurement.
We remove the data of three specimens because the wear of their diameters crosses 60 mm very quickly.
Data of the rest 11 specimens are shown in Figure \ref{wheeldata}.  From Figure
\ref{wheeldata}, we see that the degradation paths are linear and increase monotonically. Three wheels
have failed  before 600 kkm. By the linear interpolation method, we compute the failure time of
the fifth, ninth and eleventh wheels:  523.537,  558.861 and 421.508 kkm, respectively.

We fit the data by gamma process with heterogeneous effects, and RULs of the three wheels are predicted online
from the second measurement. The results are shown in Figure \ref{wheelrul}, where
the true RULs, th point predictions, and 95\% predictive intervals of RULs of the three wheels
are reported. We can see in Figure \ref{wheelrul} that the RUL predictions are reasonable close to the true RULs,
and the 95\% credible intervals cover the true RULs  at all the measurement points. This
indicates that the uncertainty quantification of the proposed algorithm is satisfactory.

\begin{figure}[htp]
	\centering
	\includegraphics[width=13cm]{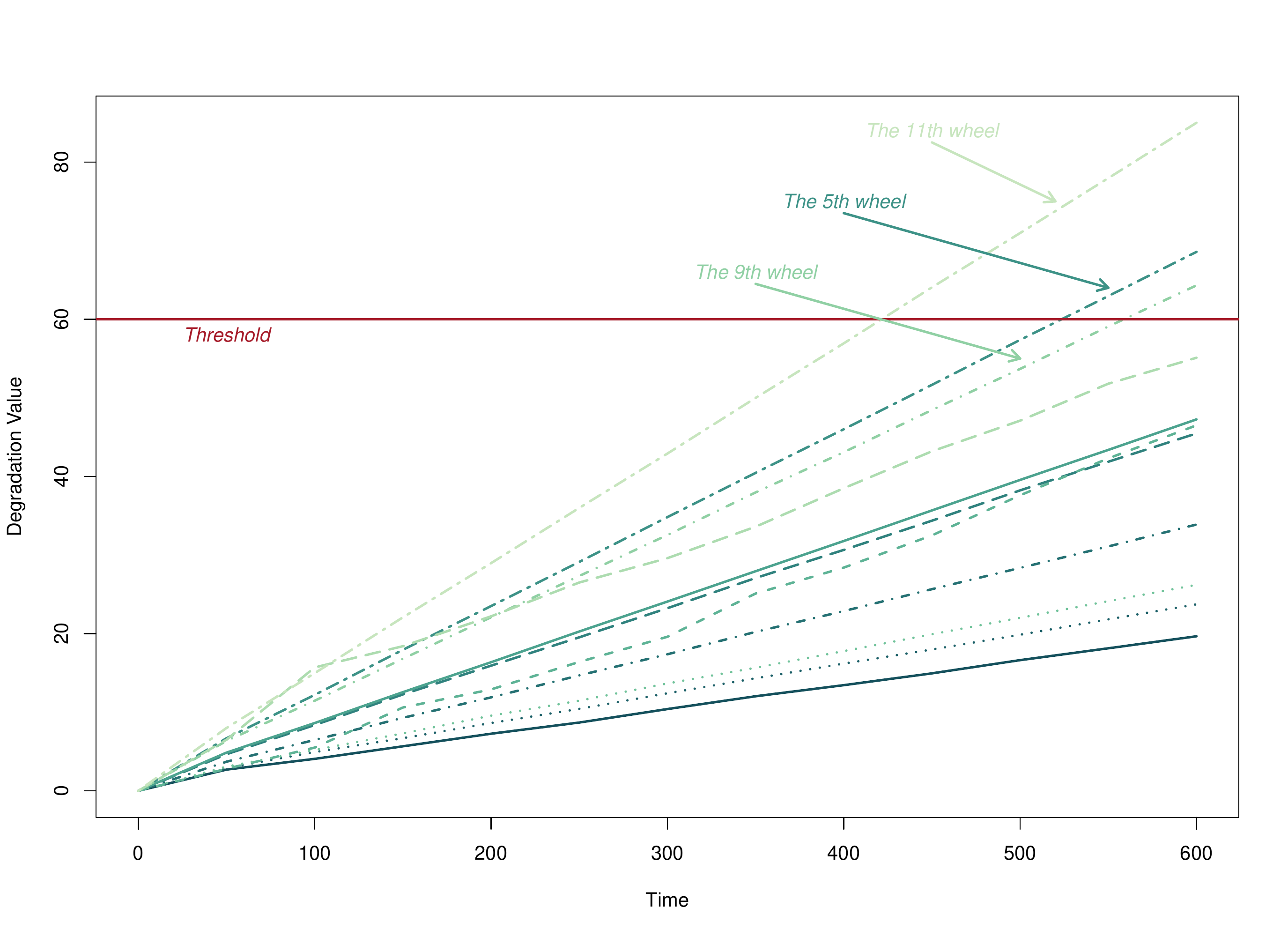}
	\caption{The train wheel data.}\label{wheeldata}
\end{figure}

\begin{figure}[htp]
	\centering
	\includegraphics[width=15cm]{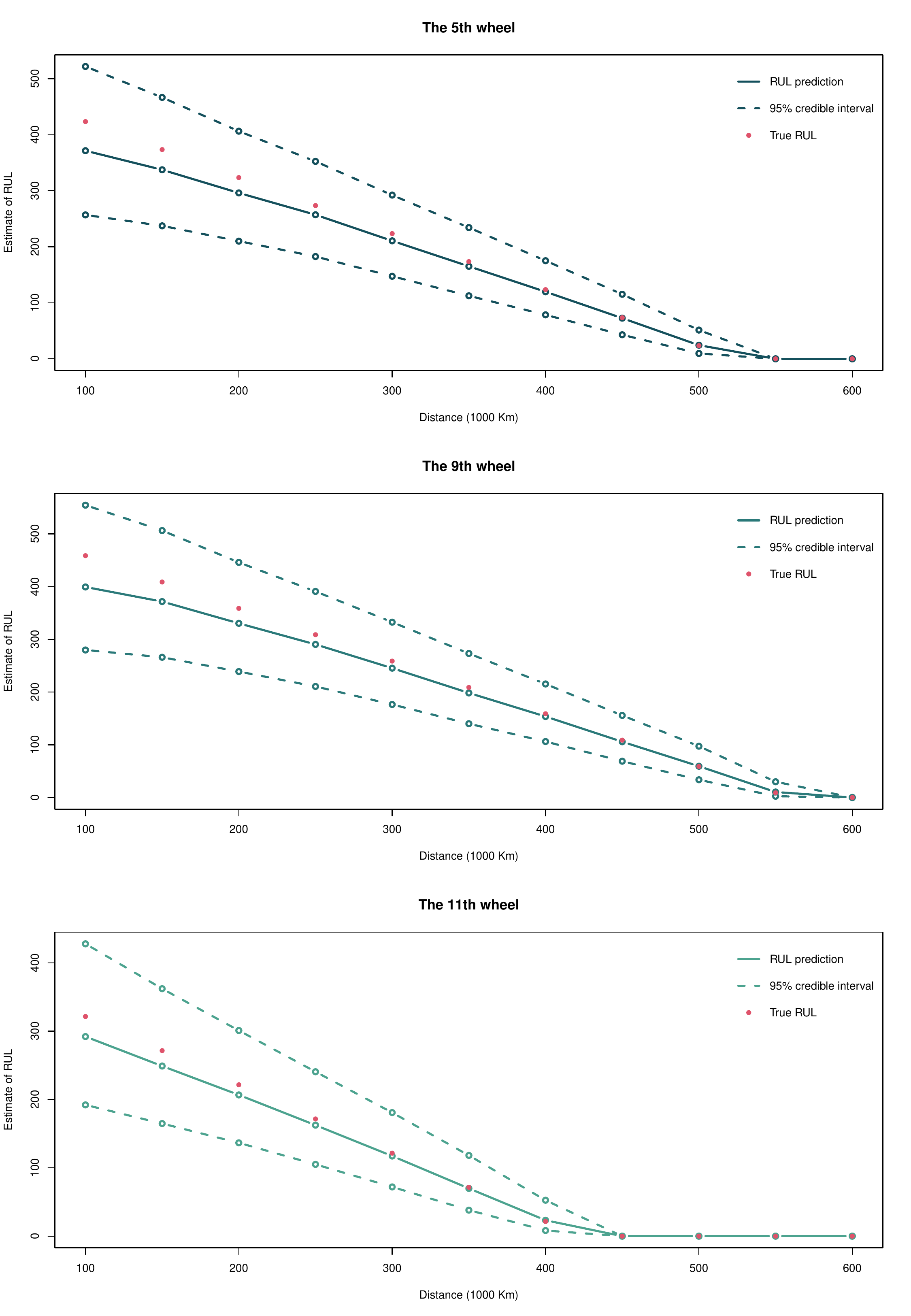}
	\caption{The point predictions and 95\% predictive intervals of RULs for the fifth, ninth and eleventh wheel.}\label{wheelrul}
\end{figure}

\section{Conclusion}
\label{sec:conc}
In this paper, we have derived a conjugate prior for the homogeneous gamma process, and
some properties of the prior distribution are studied in depth. Based on these properties,
three
algorithms (Gibbs sampling, DGS and SIR) are proposed to simulate random numbers from the posterior distribution. The generated samples can then be used to perform posterior inference. Simulation studies show that
DGS and SIR have both high computational efficiency and estimation accuracy.
The conjugate prior has been extended to the case of the gamma process with heterogeneous effects.
Similar algorithms can also be designed to generate posterior samples of the parameters.
The main advantage of a conjugate parameter structure is that the posterior distribution
can be easily updated recursively, which saves a lot of storage space and
has a high computational efficiency. With the recursive update of the posterior
distribution, an online algorithm is developed to predict the RUL of multiple systems
simultaneously. Finally,  two real-world examples have been used to validate  the proposed online algorithm, in which both point prediction and 95\% credible interval of RUL  can provide
reasonably accurate results.

\section*{Acknowledgment}
%We sincerely thank the editor, the AE and four reviewers for their insightful comments that have considerably improved the paper.
The research is supported by Natural Science Foundation of China (12171432, 11671303), the characteristic \& preponderant discipline of key construction universities  in Zhejiang province (Zhejiang Gongshang University- Statistics), and Collaborative Innovation Center of Statistical Data  Engineering Technology \& Application.

\appendix
%\section*{Appendix}		
\section{Proof of $\log\left(\frac{\prod_{j=1}^{m}t_j^{t_j/T_m}}{\overline{T}_m}\right)\ge 0$}\label{ap2}
	
Notice that $T_m=\sum_{j=1}^{m}t_j$, we have
\begin{equation*}
	\begin{aligned}
		\log\left(\frac{\prod_{j=1}^{m}t_j^{t_j/T_m}}{\overline{T}_m}\right)&=
		\sum_{j=1}^{m}\frac{t_j}{T_m}\log(t_j)-\log\left(\overline{T}_m\right)	
		=\frac{\sum_{j=1}^{m}t_j\log(t_j)}{\sum_{j=1}^{m}t_j}-\log\left(\frac{1}{m}\sum_{j=1}^{m}t_j\right)	
	\end{aligned}
\end{equation*}	
Let $q(x)=x\log(x)$. Then $q(x)$ is convex. Using Jensen's inequality, we have
$$\sum_{j=1}^{m}t_j\log(t_j)\ge m\cdot \frac{1}{m}\sum_{j=1}^{m}t_j\cdot \log\left(\frac{1}{m}\sum_{j=1}^{m}t_j\right)
=\left(\sum_{j=1}^{m}t_j\right)\cdot \log\left(\frac{1}{m}\sum_{j=1}^{m}t_j\right). $$
Thus, $\log\left(\frac{\prod_{j=1}^{m}t_j^{t_j/T_m}}{\overline{T}_m}\right)\ge 0$.

%\vspace{5mm}

\bibliographystyle{elsarticle-harv}
\bibliography{references}

\end{document}